\documentclass[twoside,leqno,letterpaper,11pt]{article}

\usepackage{amssymb,amsmath,amsthm,amscd,latexsym,algorithm}
\usepackage{url}
\usepackage{luca}
\usepackage{times}
\usepackage{paralist}
\usepackage{fullpage}

\newtheorem{theorem}{Theorem}
\newtheorem{lemma}{Lemma}
\newtheorem{corollary}{Corollary}
\theoremstyle{definition}
\newtheorem{remark}{Remark}
\newtheorem{example}{Example}

\urldef\emailk\url{krish.chat@ist.ac.at}
\urldef\emailm\url{monika.henzinger@univie.ac.at}

\def\qed{\rule{0.4em}{1.4ex}}

\newcommand{\pat}{\omega}
\newcommand{\Paths}{\Omega}
\newcommand{\PA}{1}

\newcommand{\straa}{\sigma}
\newcommand{\Straa}{\Sigma}
\newcommand{\strab}{\pi}

\newcommand{\SA}{S_1}

\newcommand{\SR}{S_{P}}

\newcommand{\gamegraph}{G}

\newcommand{\winas}[1]{\langle \! \langle #1 \rangle\! \rangle_{\mathit{almost}} }

\newcommand{\waa}{\winas{1}}

\newcommand{\Prb}{\mathrm{Pr}}
\newcommand{\Inf}{\mathrm{Inf}}

\newcommand{\Parity}{{\mathrm{Parity}}}

\newcommand{\Buchi}{\textrm{B\"uchi}}

\newcommand{\attr}{\mathit{Attr}}

\newcommand{\Nats}{\mathbb{N}}
\newcommand{\nats}{\mathbb{N}}

\newcommand{\set}[1]{\{\: #1 \:\}}
\newcommand{\seq}[1]{\langle #1 \rangle}

\newcommand{\trans}{\delta}
\newcommand{\distr}{{\cal D}}

\newcommand{\PM}{\mathit{PM}}

\newcommand{\Aa}{{\cal A}}

\newcommand{\obdd}{{\sc Obdd}}

\newcommand{\slopefrac}[2]{\leavevmode\kern.1em
  \raise .5ex\hbox{\the\scriptfont0 #1}\kern-.1em
  /\kern-.15em\lower .25ex\hbox{\the\scriptfont0 #2}}
\newcommand{\half}{\slopefrac{1}{2}}

\begin{document}

\title{Symbolic Algorithms for Qualitative Analysis of \\ Markov Decision Processes with B\"uchi Objectives
\thanks{The research was supported by Austrian Science Fund (FWF) Grant No P 23499-N23 on Modern Graph Algorithmic 
Techniques in Formal Verification, FWF NFN Grant No S11407-N23 (RiSE), ERC Start grant 
(279307: Graph Games), and Microsoft faculty fellows award.}%%$^{,}$
\thanks{Appeared in Formal Methods in System Design, 42(3):301-327, 2013. A preliminary version of the paper appeared in the proceedings of the 23rd International Conference on Computer Aided Verification, pp. 260-276, 2011.} 
}
\author{Krishnendu Chatterjee\\IST Austria\\\small{\tt krish.chat@gmail.com}  \and 
	Monika Henzinger\\University of Vienna\\\small{\tt monika.henzinger@univie.ac.at} \and
	Manas Joglekar\\Stanford University\\\small{\tt manasrj@stanford.edu} \and
	Nisarg Shah\\Carnegie Mellon University\\\small{\tt nkshah@cs.cmu.edu}
}
\date{}
%  \texttt{$\{$krishnendu.chattarjee,tah,aradha$\}$@ist.ac.at, barbara.jobstmann@imag.fr}

%\institute{IST Austria (Institute of Science and Technology Austria)
%\and CNRS/Verimag, France}

\maketitle
%\begin{titlepage}
%  \begin{center}
%    \huge{ Environment Assumptions for Synthesis}
%    \vspace{3cm}
%    
%    \begin{Large}
%      
%      \emph{Krishnendu Chatterjee$^1$}
%    
%      \emph{Thomas A. Henzinger$^2$}
%    
%      \emph{Barbara Jobstmann$^2$}
%      \vspace{1cm}
%
%      $^1$ University of California, Santa Cruz
%
%      $^2$ Ecole Polytechnique Federale de Lausanne
%
%      \vspace{10cm}
%      Technical report
%
%      \today
%      
%    \end{Large}
%
%  \end{center}
%
%\end{titlepage}

\pagenumbering{arabic}
\pagestyle{plain}

\maketitle

%\vspace{-2em}
\begin{abstract} 
We consider Markov decision processes (MDPs) with B\"uchi (liveness) 
objectives. 
%%$\omega$-regular specifications given as parity objectives. 
We consider the problem of computing the set of \emph{almost-sure} 
winning states from where the objective can be ensured with probability~1.
%%The algorithms for the computation of the almost-sure winning set 
%%for parity objectives iteratively use the solutions for the almost-sure
%%winning set for B\"uchi objectives (a special case of parity objectives). 
Our contributions are as follows: 
First, we present the first subquadratic symbolic algorithm to compute the 
almost-sure winning set for MDPs with B\"uchi objectives; our algorithm 
takes $O(n \cdot \sqrt{m})$ symbolic steps as compared to the previous
known algorithm that takes $O(n^2)$ symbolic steps, where $n$ is the number of 
states and $m$ is the number of edges of the MDP. 
In practice MDPs have constant out-degree, and then our symbolic algorithm takes
$O(n \cdot \sqrt{n})$ symbolic steps, as compared to the previous known 
$O(n^2)$ symbolic steps algorithm.
Second, we present a new algorithm, namely \emph{win-lose} algorithm, with the
following two properties: (a) the algorithm iteratively computes subsets of the 
almost-sure winning set and its complement, as compared to all previous algorithms 
that discover the almost-sure winning set upon termination; and (b) requires  
$O(n \cdot \sqrt{K})$ symbolic steps, where $K$ is the maximal number of edges of 
strongly connected components (scc's) of the MDP.
The win-lose algorithm requires symbolic computation of scc's.
Third, we improve the algorithm for symbolic scc computation; the previous 
known algorithm takes linear symbolic steps, and our new algorithm 
improves the constants associated with the linear number of steps. 
In the worst case the previous known algorithm takes $5\cdot n$ symbolic 
steps, whereas our new algorithm takes $4 \cdot n$ symbolic steps.
%%Our experimental results show that our new algorithms outperforms the previous 
%%known algorithms for computation of almost-sure winning set for B\"uchi objectives.
\end{abstract}

%\vspace{-2.0em}
\section{Introduction}
%\vspace{-1em}
\noindent{\bf Markov decision processes.} 
The standard model of systems in verification of probabilistic systems is  
\emph{Markov decision processes (MDPs)} that exhibit both probabilistic 
and nondeterministic behavior~\cite{Howard}.
MDPs have been used to model and solve control problems for stochastic systems~\cite{FV97}: 
there, nondeterminism represents the freedom of the controller to choose a 
control action, while the probabilistic component of the behavior describes the 
system response to control actions. 
MDPs have also been adopted as models for concurrent probabilistic 
systems~\cite{CY95},  probabilistic systems operating in open environments~\cite{SegalaT}, 
and under-specified probabilistic systems~\cite{BdA95}.
A \emph{specification} describes the set of desired behaviors of
the system, which in the verification and control of stochastic systems is 
typically an $\omega$-regular set of paths. % in the MDP. 
The class of $\omega$-regular languages extends classical regular languages to 
infinite strings, and provides a robust specification language to express
all commonly used specifications, such as safety, liveness, fairness, etc.~\cite{Thomas97}. 
Parity objectives are a canonical way to define such $\omega$-regular specifications.
Thus MDPs with parity objectives provide the theoretical framework to 
study problems such as the verification and control of stochastic systems.

\smallskip\noindent{\bf Qualitative and quantitative analysis.} 
The analysis of MDPs with parity objectives can be classified into  
qualitative and quantitative analysis. 
Given an MDP with parity objective, the \emph{qualitative analysis} 
asks for the computation of the set of states from
where the parity objective can be ensured with probability~1 
(almost-sure winning).
%%in which the controller can achieve the parity objective with probability~1 
%%(almost-surely).
The more general \emph{quantitative analysis} asks for the computation 
of the maximal probability at each state with which the controller can 
satisfy the parity objective. 

\smallskip\noindent{\bf Importance of qualitative analysis.} 
The qualitative analysis of MDPs is an important problem in verification 
that is of interest irrespective of the quantitative analysis problem.
There are many applications where we need to know whether the correct 
behavior arises with probability~1.
For instance, when analyzing a randomized embedded scheduler, we are
interested in whether every thread progresses with probability~1
\cite{EMSOFT05}.
Even in settings where it suffices to satisfy certain specifications with 
probability $p<1$, the correct choice of $p$ is a challenging problem, due 
to the simplifications introduced during modeling.
For example, in the analysis of randomized distributed algorithms it is 
quite common to require correctness with probability~1 
(see, e.g., \cite{PSL00,KNP_PRISM00,Sto02b}). 
Furthermore, in contrast to quantitative analysis, 
qualitative analysis is robust to numerical perturbations and modeling errors in the 
transition probabilities, and consequently the algorithms for qualitative analysis are 
combinatorial.
%%as compared to the numerical algorithms for the quantitative analysis. 
Finally, for MDPs with parity objectives, the best known algorithms and all
algorithms used in practice first perform the qualitative analysis, and then 
perform a quantitative analysis on the result of the qualitative 
analysis~\cite{CY95,luca-thesis,CJH04}. 
Thus qualitative analysis for MDPs with parity objectives is one of the most 
fundamental and core problems in verification of probabilistic systems.
One of the key challenges in probabilistic verification is to obtain efficient 
and symbolic algorithms for qualitative analysis of MDPs with parity 
objectives, as symbolic algorithms allow to handle MDPs with a large state space.

\smallskip\noindent{\bf Previous results.} The qualitative analysis for MDPs with 
parity objectives is achieved by iteratively applying solutions of the 
qualitative analysis of MDPs with B\"uchi objectives~\cite{CY95,luca-thesis,CJH04}.
The qualitative analysis of an MDP with a parity objective with $d$ priorities 
can be achieved  by $O(d)$ calls to an algorithm for qualitative analysis of MDPs with B\"uchi
objectives, and hence we focus on the qualitative analysis of MDPs with B\"uchi
objectives. 
The classical algorithm for qualitative analysis for MDPs with B\"uchi objectives
works in $O(n \cdot m)$ time, where $n$ is the number of states, and $m$ is the number of
edges of the MDP~\cite{CY95,luca-thesis}. 
The classical algorithm can be implemented symbolically, and it takes at most 
$O(n^2)$ symbolic steps. 
An improved algorithm for the problem was given in~\cite{CJH03} that works in 
$O(m \cdot \sqrt{m})$ time. 
The algorithm of~\cite{CJH03} crucially depends on maintaining the same number 
of edges in certain forward searches. Thus the algorithm needs to explore edges 
of the graph explicitly and is inherently non-symbolic.
A recent $O(m\cdot n^{2/3})$ time algorithm for the problem was given in~\cite{CMoH11};
however the algorithm requires the dynamic-tree data structure of Sleator-Tarjan~\cite{SleatorTarjan}, 
and such data structures cannot be implemented symbollically.
In the literature, there is no symbolic subquadratic algorithm for qualitative analysis of
MDPs with B\"uchi objectives.

\smallskip\noindent{\bf Our contribution.} 
In this work our main contributions are as follows.
\begin{enumerate}
\item We present a new and simpler subquadratic algorithm for qualitative analysis 
of MDPs with B\"uchi objectives that runs in $O(m \cdot \sqrt{m})$ time, and show 
that the algorithm can be implemented symbolically. 
The symbolic algorithm takes at most $O(n \cdot \sqrt{m})$ symbolic steps, and thus
we obtain the first symbolic subquadratic algorithm.
In practice, MDPs often have constant out-degree: for example, see~\cite{dAR07} for MDPs with large
state space but constant number of actions, or~\cite{FV97,Puterman94} for examples from inventory 
management where MDPs have constant number of actions (the number of actions correspond to 
the out-degree of MDPs).
For MDPs with constant out-degree our new symbolic algorithm takes 
$O(n \cdot \sqrt{n})$ symbolic steps, as compared to $O(n^2)$ symbolic steps of the 
previous best known algorithm.

\item All previous algorithms for the qualitative analysis of MDPs  with B\"uchi 
objectives iteratively discover states that are guaranteed to be not almost-sure 
winning, and only when the algorithm terminates the almost-sure winning set is 
discovered. We present a new algorithm (namely \emph{win-lose} algorithm) that iteratively 
discovers both states in the almost-sure winning set and its complement. 
Thus if the problem is to decide whether a given state $s$ is almost-sure 
winning, and the state $s$ is almost-sure winning, then the win-lose algorithm 
can stop at an intermediate iteration unlike all the previous algorithms. 
Our algorithm works in time $O( \sqrt{K_E} \cdot m)$ time, where $K_E$ is 
the maximal number of edges of any scc of the MDP 
(in this paper we write scc for maximal scc).
We also show that the win-lose algorithm can be implemented symbolically, and it takes at most
$O(\sqrt{K_E} \cdot n)$ symbolic steps.

\item  Our win-lose algorithm requires to compute the scc decomposition of a 
graph in $O(n)$ symbolic steps. 
The scc decomposition problem is one of the most fundamental problem 
in the algorithmic study of graph problems.
The symbolic scc decomposition problem has many other applications in 
verification: 
for example, checking emptiness of $\omega$-automata, and bad-cycle detection 
problems in model checking, see~\cite{BGS00} for other applications.
An $O(n \cdot \log n)$ symbolic step algorithm for scc decomposition was 
presented in~\cite{BGS00}, and the algorithm was improved in~\cite{GPP}.
The algorithm of~\cite{GPP} is a linear symbolic step scc decomposition algorithm that 
requires at most $\min\set{5\cdot n, 5 \cdot D \cdot N + N}$ symbolic steps, where $D$ is 
the diameter of the graph, and $N$ is the number of scc's of the graph.
We present an improved version of the symbolic scc decomposition algorithm.
Our algorithm improves the constants of the number of the linear symbolic steps.
Our algorithm requires at most 
$\min\set{ 3\cdot n + N, 5 \cdot D^* + N}$ symbolic steps, where $D^*$ is 
the sum of the diameters of the scc's of the graph.
Thus, in the worst case, the algorithm of~\cite{GPP} requires $5\cdot n$ symbolic steps,
whereas our algorithm requires $4\cdot n$ symbolic steps.
Moreover, the number of symbolic steps of our algorithm is always bounded by
the number of symbolic steps of the algorithm of~\cite{GPP} (i.e. our algorithm
is never worse).
\end{enumerate}
Our experimental results show that our new algorithms perform better than 
the previous known algorithms both for qualitative analysis of MDPs with 
B\"uchi objectives and symbolic scc computation.

%\vspace{-1em}
\section{Definitions}
\label{section:definition}
%\vspace{-1em}

\noindent{\bf Markov decision processes (MDPs).}
A \emph{Markov decision process (MDP)} $\gamegraph =((S, E), (\SA,\SR),\trans)$ 
consists of a directed graph $(S,E)$, a partition $(\SA$,$\SR)$ of the 
\emph{finite} set $S$ of states, and a probabilistic transition function 
$\trans$: $\SR \rightarrow \distr(S)$, where $\distr(S)$ denotes the 
set of probability distributions over the state space~$S$. 
The states in $\SA$ are the {\em player-$\PA$\/} states, where player~$\PA$
decides the successor state, and the states in $\SR$ are the 
{\em probabilistic (or random)\/} states, 
where the successor state is chosen according to the probabilistic transition
function~$\trans$. 
We assume that for $s \in \SR$ and $t \in S$, we have $(s,t) \in E$ 
iff $\trans(s)(t) > 0$, and we often write $\trans(s,t)$ for $\trans(s)(t)$. 
For a state $s\in S$, we write $E(s)$ to denote the set 
$\set{t \in S \mid (s,t) \in E}$ of possible successors.
For technical convenience we assume that every state in the graph 
$(S,E)$ has at least one outgoing edge, i.e., $E(s)\neq \emptyset$ for all $s \in S$.
We will denote by $n=|S|$ and $m=|E|$ the size of the state space and the number of 
transitions (or edges), respectively.
%
%%The {\em turn-based deterministic game graphs} (\emph{2-player game graphs})
%%are the special case of the $2\half$-player game graphs with $\SR = \emptyset$.
%%The \emph{Markov decision processes} (\emph{$1\half$-player game graphs}) 
%%are the special case of the $2\half$-player game graphs with 
%%$\SA = \emptyset$ or $\SB = \emptyset$. 
%%We refer to the MDPs with $\SB=\emptyset$ as \emph{player-$\PA$} MDPs,
%%and to the MDPs with $\SA=\emptyset$ as \emph{player-$\PB$} MDPs.

\smallskip\noindent{\bf Plays and strategies.}
An infinite path, or a \emph{play}, of the game graph $\gamegraph$ is an 
infinite 
sequence $\pat=\seq{s_0, s_1, s_2, \ldots}$ of states such that 
$(s_k,s_{k+1}) \in E$ for all $k \in \Nats$. 
We write $\Paths$ for the set of all plays, and for a state $s \in S$, 
we write $\Paths_s\subseteq\Paths$ 
for the set of plays that start from the state~$s$.
A \emph{strategy} for  player~$\PA$ is a function 
$\straa$: $S^*\cdot \SA \to \distr(S)$ that chooses the probability 
distribution over the successor states 
for all finite sequences $\vec{w} \in S^*\cdot \SA$ of states 
ending in a player-1 state (the sequence represents a prefix of a play).
A strategy must respect the edge relation: for all $\vec{w} \in S^*$ and 
$s \in \SA$, if $\straa(\vec{w}\cdot s)(t) >0$, then $t \in E(s)$.
A strategy is \emph{deterministic (pure)} if it chooses a unique successor
for all histories (rather than a probability distribution), otherwise
it is \emph{randomized}.
Player~$\PA$ follows the strategy~$\straa$ if in each player-1 
move, given that the current history of the game is
$\vec{w} \in S^* \cdot \SA$, she chooses the 
next state according to  $\straa(\vec{w})$.
We denote by $\Straa$ the set of all strategies for player~$\PA$.
A \emph{memoryless} player-1 strategy does not depend on the history of 
the play but only on the current state; i.e., for all $\vec{w},\vec{w'} \in 
S^*$ and for all $s \in \SA$ we have 
$\straa(\vec{w} \cdot s) =\straa(\vec{w}'\cdot s)$.
A memoryless strategy can be represented as a function 
$\straa$: $\SA \to \distr(S)$, and a pure memoryless 
strategy can be represented as $\straa: \SA \to S$.
%%We denote by $\Straa^{M}$ the set of  memoryless strategies for player~1.

Once a starting state  $s \in S$ and a strategy $\straa \in \Straa$ is fixed, 
the outcome of the MDP is a random walk $\pat_s^{\straa}$ for which the
probabilities of events are uniquely defined, where an \emph{event}  
$\Aa \subseteq \Paths$ is a measurable set of plays. 
For a state $s \in S$ and an event $\Aa\subseteq\Paths$, we write
$\Prb_s^{\straa}(\Aa)$ for the probability that a play belongs 
to $\Aa$ if the game starts from the state $s$ and player~1 follows
the strategy $\straa$.
%%For a measurable function $f:\Paths \to \reals$ we denote by 
%%$\Exp_s^{\straa,\strab}[f]$ the \emph{expectation} of the function
%%$f$ under the probability measure $\Prb_s^{\straa,\strab}(\cdot)$.

\begin{comment}
Given a pure memoryless strategy $\straa \in \Straa^{\PM}$, let 
$G_\straa$ be the game graph obtained from $G$ under the 
constraint that player~1 follows the strategy~$\sigma$.
The corresponding definition $G_\pi$ for a player-2 strategy 
$\pi\in\Pi^{\PM}$ is analogous, and we write $G_{\sigma,\pi}$ 
for the game graph obtained from $G$ if both players follow the 
pure memoryless strategies $\sigma$ and~$\pi$, respectively.
Observe that given a $2\half$-player game graph $G$ and a pure memoryless 
player-1 strategy~$\straa$, the result $G_{\straa}$ is a player-2 MDP. 
Similarly, for a player-1 MDP $G$ and a pure memoryless
player-1 strategy~$\straa$, the result $G_\straa$ is a Markov chain. 
Hence, if $G$ is a $2\half$-player game graph and the two players follow 
pure memoryless strategies $\straa$ and~$\strab$, the result
$G_{\straa,\strab}$ is a Markov chain. 
\end{comment}

\smallskip\noindent{\bf Objectives.}
We specify \emph{objectives} for the player~1 by providing
a set of \emph{winning} plays $\Phi \subseteq \Omega$.
We say that a play $\pat$ {\em satisfies} the objective
$\Phi$ if $\pat \in \Phi$.
We consider \emph{$\omega$-regular objectives}~\cite{Thomas97},
specified as parity conditions.
We also consider the special case of B\"uchi objectives.
%%We also define reachability objectives, which is an important special 
%%class of $\omega$-regular objectives.
\begin{itemize}
%%\item
%%  \emph{Reachability objectives.}
%%  Given a set $T\subseteq S$ of ``target'' states, the reachability
%%  objective requires that some state of $T$ be visited.
%%  The set of winning plays is 
%%  $\Reach(T) =
%%  \set{\seq{s_0, s_1, s_2,\ldots} \in \Paths \mid
%%    s_k \in T \mbox{ for some }k \ge 0}$.

\item
 \emph{B\"uchi objectives.} Let $T$ be a set of target states.
  For a play $\pat = \seq{s_0, s_1, \ldots} \in \Omega$,
  we define $\Inf(\pat) =
  \set{s \in S \mid \mbox{$s_k = s$ for infinitely many $k$}}$
  to be the set of states that occur infinitely often in~$\pat$.
  The B\"uchi objectives require that some state of $T$ be visited 
  infinitely often, and defines the set of winning plays 
  $\Buchi(T)=\set{\pat \in \Paths \mid \Inf(\pat) \cap T \neq \emptyset}$.

\item
  \emph{Parity objectives.}
  For $c,d \in \nats$, we write $[c..d] = \set{c, c+1, \ldots, d}$.
  Let $p$: $S \to [0..d]$ be a function that assigns a \emph{priority}
  $p(s)$ to every state $s \in S$, where $d \in \Nats$.
  The \emph{parity objective} is defined as
  $\Parity(p)=
  \set{\pat \in \Paths \mid
  \min\big(p(\Inf(\pat))\big) \text{ is even }}$.
  In other words, the parity objective requires that the minimum 
  priority visited infinitely often is even.
  In the sequel we will use $\Phi$ to denote parity objectives.
\end{itemize}

\smallskip\noindent{\em Qualitative analysis: almost-sure winning.}
Given a player-1 objective~$\Phi$, a strategy $\straa\in\Sigma$ is  
\emph{almost-sure winning} for player~1 from the state $s$ 
if $\Prb_s^{\straa} (\Phi) =1$.
The \emph{almost-sure winning set} $\waa(\Phi)$ for player~1 is the set 
of states from which player~1 has an almost-sure winning strategy.
The qualitative analysis of MDPs correspond to the computation of 
the almost-sure winning set for a given objective $\Phi$.
It follows from the results of~\cite{CY95,luca-thesis} that for all MDPs and 
all reachability and parity objectives, if there is an almost-sure 
winning strategy, then there is a memoryless almost-sure winning strategy.
The qualitative analysis of MDPs with parity objectives is achieved by 
iteratively applying the solutions of qualitative analysis for MDPs with 
B\"uchi objectives~\cite{luca-thesis,CJH04}, and hence in this 
work we will focus on qualitative analysis for B\"uchi objectives.

%\vspace{-0.5em}
\begin{theorem}[\cite{CY95,luca-thesis}]\label{thrm_memless}
For all MDPs $G$, and all reachability and parity objectives $\Phi$, 
there exists a pure memoryless strategy $\straa_*$ such that for all 
$s \in \waa(\Phi)$ we have $\Prb_s^{\straa_*}(\Phi)=1$.
\end{theorem}
%\vspace{-0.5em}

\noindent{\bf Scc and bottom scc.} Given a graph $G=(S,E)$, a set $C$ 
of states is an scc if for all $s,t \in C$ there is a path from 
$s$ to $t$ going through states in $C$. 
An scc $C$ is a bottom scc if for all $s \in C$ all out-going edges 
are in $C$, i.e., $E(s) \subseteq C$.

\smallskip\noindent{\bf Markov chains, closed recurrent sets.} A Markov chain is a 
special case of MDP with $\SA=\emptyset$, and hence for simplicity a Markov 
chain is a tuple $((S,E),\trans)$ with a probabilistic transition function 
$\trans:S \to \distr(S)$, and $(s,t) \in E$ iff $\trans(s,t)>0$.
A \emph{closed recurrent} set $C$ of a Markov chain is a bottom scc 
in the graph $(S,E)$.
Let $\calc= \bigcup_{C \text{ is closed recurrent}} C$. 
It follows from the results on Markov chains~\cite{Kemeny} that for all 
$s \in S$, the set $\calc$ is reached with probability~1 in finite time,
and for all $C$ such that $C$ is closed recurrent, for all $s \in C$ and 
for all $t \in C$, if the starting state is $s$, then the state $t$ is 
visited infinitely often with probability~1.

\smallskip\noindent{\bf Markov chain from a MDP and memoryless strategy.} 
Given a MDP $G=((S,E),(\SA,\SR),\trans)$ and a memoryless strategy 
$\straa_*:\SA \to \distr(S)$ we obtain a Markov chain $G'=((S,E'),\trans')$ as follows:
$E' = E \cap (\SR \times S) \cup \set{(s,t) \mid s \in \SA, \straa_*(s)(t) > 0}$;
and 
$\trans'(s,t)=\trans(s,t)$ for $s \in \SR$, and 
$\trans'(s,t)=\straa(s)(t)$  for $s\in \SA$ and $t \in E(s)$. 
We will denote by $G_{\straa_*}$ the Markov chain obtained from an MDP $G$ 
by fixing a memoryless strategy $\straa_*$ in the MDP.

\smallskip\noindent{\bf Symbolic encoding of an MDP.} All algorithms of the 
paper will only depend on the graph $(S,E)$ of the MDP and the partition 
$(\SA,\SR)$, and not on the probabilistic transition function $\trans$. 
Thus the symbolic encoding of an MDP is obtained as the standard encoding 
of a transition system (with an \obdd~\cite{CuDD}), with one additional bit, 
and the bit denotes whether a state belongs to $\SA$ or $\SR$. 
Also note that if the state bits already encode whether a state belongs to 
$\SA$ or $\SR$, then the additional bit is not required.

%%%%% KRISH_J1: CHECK START
\smallskip\noindent{\bf Symbolic step.} To define the symbolic complexity of
an algorithm an important concept to clarify is the notion of one symbolic
step. 
In this work we adopt the following convention: one symbolic step corresponds
to one primitive operations that are supported by the standard symbolic package like CuDD~\cite{CuDD}.
For example, the one-step predecessor and successor operators, obtaining a BDD for a cube 
(a path from root to a leaf node with constant~1) of a BDD, etc. are all supported as 
primitive operations in CuDD~\cite{CuDD} and correspond to one symbolic step.
%%%%% KRISH_J1: CHECK END

\newcommand{\impa}{{\sc ImprAlgo}}
\newcommand{\simpa}{{\sc SymbImprAlgo}}
\newcommand{\oimpa}{{\sc SmDvSymbImprAlgo}}
\newcommand{\classical}{{\sc Classical}}
\newcommand{\Pre}{\mathsf{Pre}}
\newcommand{\pre}{\mathsf{Pre}}
\newcommand{\post}{\mathsf{Post}}
\newcommand{\CPre}{\mathsf{CPre}}
\newcommand{\Post}{\mathsf{Post}}
\newcommand{\BDD}{\mathsf{BDD}}
\newcommand{\SymbStep}{\mathsf{SymbStep}}

\section{Symbolic Algorithms for B\"uchi Objectives}\label{sec:symbolic-buchi}
In this section we will present a new improved algorithm for the qualitative 
analysis of MDPs with B\"uchi objectives, and then present a symbolic 
implementation of the algorithm. 
Thus we obtain the first symbolic subquadratic algorithm for the problem.
We start with the notion of \emph{attractors} that is crucial for our 
algorithm.

\smallskip\noindent{\bf Random and player~1 attractor.} 
Given an MDP $G$, let $U \subseteq S$ be a subset of states. 
The \emph{random attractor} $\attr_{R}(U)$ is defined inductively as follows: 
$X_0=U$, and for $i\geq 0$, let  
$X_{i+1}=X_i \cup \set{s \in \SR \mid E(s) \cap X_i \neq \emptyset } \cup 
\set{s \in \SA \mid E(s) \subseteq X_i}$.
In other words, $X_{i+1}$ consists of (a)~states in $X_i$, 
(b)~player-1 states whose all successors are in $X_i$ and 
(c)~random states that have at least one edge to $X_i$.
Then $\attr_{R}(U)=\bigcup_{i\geq 0} X_i$. 
The definition of \emph{player-1 attractor} $\attr_1(U)$ is analogous 
and is obtained by exchanging the role of random states and 
player~1 states in the above definition.

\smallskip\noindent{\bf Property of attractors.} 
Given an MDP $G$, and set $U$ of states, let $A=\attr_{R}(U)$. 
Then from $A$ player~1 cannot ensure to avoid $U$, in other words, for 
all states in $A$ and for all player~1 strategies, the set $U$ is 
reached with positive probability.
For $A=\attr_1(U)$ there is a player~1 memoryless strategy to ensure 
that the set $U$ is reached with certainty.
The computation of random and player~1 attractors is the computation of 
alternating reachability and can be achieved in $O(m)$ time~\cite{Immerman81}, 
and can be achieved in $O(n)$ symbolic steps.

\subsection{A new subquadratic algorithm} 
The classical algorithm for computing the almost-sure winning set in 
MDPs with B\"uchi objectives has $O(n \cdot m)$ running time, and 
the symbolic implementation of the algorithm takes at most $O(n^2)$ 
symbolic steps. 
A subquadratic algorithm, with $O(m \cdot \sqrt{m})$ running time, for 
the problem was presented in~\cite{CJH03}. 
The algorithm of~\cite{CJH03} uses a mix of backward exploration and 
forward exploration. Every forward exploration step consists of executing a set
of DFSs (depth first searches) simultaneously for a specified number of 
edges, and must maintain the exploration of the same number of edges in each 
of the DFSs.
The algorithm thus depends  crucially on maintaining the number of edges 
traversed explicitly, and hence the algorithm has no symbolic implementation.
In this section we present a new subquadratic algorithm to compute 
$\waa(\Buchi(T))$.
The algorithm is simpler as compared to the algorithm of~\cite{CJH03} 
and we will show that our new algorithm can be implemented symbolically. 
Our new algorithm has some similar ideas as the algorithm of~\cite{CJH03} in 
mixing backward and forward exploration, 
but the key difference is that the new algorithm never stops the forward 
exploration after a certain number of edges, and hence need not maintain the
traversed edges explicitly. 
Thus the new algorithm is simpler, and our correctness and running time 
analysis proofs are different.
We show that our new algorithm works in $O(m \cdot \sqrt{m})$ time, and 
requires at most $O(n \cdot \sqrt{m})$ symbolic steps.

\smallskip\noindent{\bf Improved algorithm for almost-sure B\"uchi.}
Our algorithm iteratively removes states from the graph, until the almost-sure
winning set is computed. 
At iteration $i$, we denote the remaining subgraph as $(S_i,E_i)$, where 
$S_i$ is the set of remaining states, $E_i$ is the set of remaining 
edges, and the set of remaining target states is $T_i$ (i.e., $T_i= S_i \cap T$).
The set of states removed will be denoted by $Z_i$, i.e., $S_i=S \setminus Z_i$.
The algorithm will ensure that (a)~$Z_i \subseteq S \setminus \waa(\Buchi(T))$;
and (b)~for all $s \in S_i \cap \SR$ we have $E(s) \cap Z_i=\emptyset$.
In every iteration the algorithm identifies a set $Q_i$ of states such that
there is no path from $Q_i$ to the set $T_i$.
Hence clearly $Q_i \subseteq S \setminus \waa(\Buchi(T))$. 
By the random attractor property from $\attr_R(Q_i)$ the set $Q_i$ is reached 
with positive probability against any strategy for player~1.
The algorithm maintains %a set $I_{i+1}$ as the union of 
the set $L_{i+1}$ of states that were removed from the graph since (and including) 
the last iteration of Case~1, and the set $J_{i+1}$ of states that lost an edge to states 
removed from the graph since the last iteration of Case~1. 
Initially $L_0:= J_0:=\emptyset$, $Z_0:=\emptyset$, and let $i:=0$ and we describe the 
iteration $i$ of our algorithm. We call our algorithm \impa\ 
(Improved Algorithm) and the pseudocode is given as Algorithm~\ref{algorithm:ImprAlgo}.

\begin{enumerate}
\item \emph{Case~1.} If ($(|J_i| > \sqrt{m})$ or $i=0$), then 
\begin{enumerate}
\item Let $Y_i$ be the set of states that can reach the current target set $T_i$ (this can 
be computed in $O(m)$ time by a graph reachability algorithm).

\item Let $Q_i:= S_i \setminus Y_i$, i.e., there is no path from $Q_i$ to $T_i$.

\item $Z_{i+1}:= Z_i \cup \attr_{R}(Q_i)$. The set $\attr_{R}(Q_i)$ 
is removed from the graph.

\item The set $L_{i+1}$ is the set of states removed from the graph in
this iteration (i.e., $L_{i+1}:= \attr_{R}(Q_i)$) and $J_{i+1}$ 
be the set of states in the remaining graph with an edge to $L_{i+1}$. 
%%The set $I_{i+1}=J_{i+1} \cup L_{i+1}$. 

\item If $Q_i$ is empty, the algorithm stops, otherwise $i:=i+1$ and go to the next iteration.

\end{enumerate}

\item \emph{Case 2.} Else  ($|J_i| \leq \sqrt{m}$ and $i>0$), then 
\begin{enumerate}
%%\item Consider the set $J_{i}$ to be the set of vertices in the graph that 
%%lost an edge to the states removed since the last iteration that 
%%executed Case~1.   

\item We do a lock-step search from every state $s$ in $J_i$ as follows: 
we do a DFS from $s$ and (a)~if the DFS tree reaches a state in $T_i$, then we
stop the DFS search from $s$; and (b)~if the DFS is completed without 
reaching a state in $T_i$, then we stop the entire lock-step search, 
and all states in the DFS tree are identified as $Q_i$.
The set $\attr_{R}(Q_i)$ is removed from the graph and $Z_{i+1}:=Z_i \cup \attr_{R}(Q_i)$.
If DFS searches from all states $s$ in $J_i$ reach the set $T_i$, 
then the algorithm stops.

\item The set $L_{i+1}$ is the set of states removed from the graph since 
the last iteration of Case~1 (i.e., $L_{i+1} := L_i \cup \attr_R(Q_i)$, 
where $Q_i$ is the DFS tree that stopped without reaching $T_i$ in the previous step  
of this iteration) and $J_{i+1}$ be the set of states in the remaining graph with 
an edge to $L_{i+1}$, i.e., $J_{i+1} := (J_i \setminus \attr_{R}(Q_i)) \cup X_i$, 
where $X_i$ is the subset of states of $S_i$ with an edge to $\attr_{R}(Q_i)$. 
%%The set $I_{i+1}= J_{i+1} \cup L_{i+1}$. 

\item $i:=i+1$ and go to the next iteration.

\end{enumerate}

\end{enumerate}

\begin{algorithm}[ht]
\caption{\bf ImprAlgo}
\label{algorithm:ImprAlgo}
{ 
\begin{tabbing}
aa \= aa \= aa \= aa \= aa \= aa \= aa \= aa \= aa \= aa \= aa \= aa \kill
\> \textbf{Input}: An MDP $G=((S, E), (\SA,\SR),\trans)$  with B\"uchi set $T$. \\
\> \textbf{Output}: $\waa(\Buchi(T))$, i.e., the almost-sure winning set for player~1. \\ 
\> 1. $i := 0$; $S_{0} := S$; $E_{0} := E$; $T_{0} := T$;  \\
\> 2. $L_{0} := Z_{0} := J_{0} := \emptyset $; \\ %might need to expand reach reach is%
\> 3. \textbf{if} $(|J_{i}| > \sqrt{m}$ or $i = 0)$ \textbf{then}\\
\> \> 3.1. $Y_{i} :=$ {\bf Reach}$(T_{i}, (S_{i}, E_{i}))$; (i.e., compute the set $Y_i$ that can reach $T_i$ in the graph $(S_i,E_i)$) \\
\> \> 3.2. $Q_{i} := S_{i} \setminus Y_{i}$; \\
\> \> 3.3. \textbf{if} $(Q_{i} = \emptyset)$ \textbf{then} \textbf{goto} line $6$; \\
\> \> 3.4. \textbf{else} \textbf{goto} line $5$; \\
\> 4. \textbf{else} (i.e., $J_i \leq \sqrt{m}$ and $i>0$) \\
\> \> 4.1. \textbf{for each} $s \in J_{i}$ \\
\> \> \> 4.1.1. $\mathit{DFS}_{i,s} := s$; (initializing DFS-trees)\\
\> \> 4.2.  \textbf{for each} $s \in J_{i}$ \\
\> \> \> 4.2.1. Do 1 step of DFS from $\mathit{DFS}_{i,s}$, unless it has encountered a state from $T_{i}$\\
\> \> \> 4.2.2. If DFS encounters a state from $T_{i}$, mark that DFS as stopped \\ 
\> \> \> 4.2.3. \textbf{if} DFS completes without meeting $T_{i}$ \textbf{then} \\
\> \> \> \> 4.2.3.1. $Q_{i} := \mathit{DFS}_{i,s}$; \\
\> \> \> \> 4.2.3.2. \textbf{goto} line $5$; \\ 
\> \> \> 4.2.4. \textbf{if} all DFSs meet $T_{i}$ \textbf{then} \\
\> \> \> \> 4.2.4.1. \textbf{goto} line $6$; \\

\> 5. Removal of attractor of $Q_i$ in the following steps \\ 
\> \> 5.1 $Z_{i+1} := Z_{i} \cup \attr_{R}(Q_{i}, (S_{i},E_i),(\SA\cap S_i,\SR \cap S_i))$; \\
\> \> 5.2. $S_{i+1} := S_{i} \setminus Z_{i+1}$; $E_{i+1} := E_{i} \cap S_{i+1} \times S_{i+1}$; \\
\> \> 5.4. \textbf{if} the last goto call from step 3.4 (i.e. Case 1 is executed) \textbf{then} \\
\> \> \> 5.4.1 $L_{i+1} := \attr_{R}(Q_{i}, (S_{i},E_i),(\SA\cap S_i,\SR \cap S_i))$;\\
\> \> 5.5. \textbf{else } $L_{i+1} := L_i \cup \attr_{R}(Q_{i}, (S_{i},E_i),(\SA\cap S_i,\SR \cap S_i))$;\\
\> \> 5.6  $J_{i+1} := E^{-1}(L_{i+1}) \cap S_{i+1}$; \\
\> \> 5.7. $i := i + 1$; \\
\> \> 5.8.  \textbf{goto} line $3$; \\

\> 6. \textbf{return} $S\setminus Z_{i}$; \\
\end{tabbing}
}
\end{algorithm}

\smallskip\noindent{\bf Correctness and running time analysis.}
We first prove the correctness of the algorithm.

\begin{lemma}\label{lem:correct}
Algorithm \impa\ correctly computes the set $\waa(\Buchi(T))$. 
\end{lemma}
\begin{proof} 
We consider an iteration $i$ of the algorithm.
Recall that in this iteration $Y_i$ is the set of states that can reach 
$T_i$ and $Q_i$ is the set of states with no path to $T_i$.
Thus the algorithm ensures that in every iteration $i$, for the set of states 
$Q_i$ identified by the algorithm there is no path to the set $T_i$, 
and hence from $Q_i$ the set $T_i$ cannot be reached with positive 
probability.
Clearly, from $Q_i$ the set $T_i$ cannot be reached with probability~1.
Since from $\attr_R(Q_i)$ the set $Q_i$ is reached with positive 
probability against all strategies for player~1, it follows that from 
$\attr_{R}(Q_i)$ the set $T_i$ cannot be ensured to be reached with 
probability~1. 
Thus for the set $Z_i$ of removed states  we have 
$Z_i\subseteq S \setminus \waa(\Buchi(T))$.
It follows that all the states removed by the algorithm over all iterations are 
not part of the almost-sure winning set.

To complete the correctness argument we show that when the algorithm 
stops, the remaining set is $\waa(\Buchi(T))$. 
When the algorithm stops, let $S_*$ be the set of remaining states and 
$T_*$ be the set of remaining target states.
It follows from above that $S\setminus S_* \subseteq S \setminus \waa(\Buchi(T))$ and 
to complete the proof we show $S_* \subseteq \waa(\Buchi(T))$.
The following assertions hold: (a)~for all $s \in S_* \cap \SR$ we have
$E(s) \subseteq S_*$, and (b)~for all states $s \in S_*$ there is a path 
to the set $T_*$. 
We prove (a) as follows: whenever the algorithm removes a set $Z_i$, it is a random attractor,
and thus if a state $s \in S_* \cap \SR$ has an edge $(s,t)$ with $t \in S \setminus S_*$, 
then $s$ would have been included in $S \setminus S_*$, and thus (a) follows.
We prove (b) as follows: (i)~If the algorithm stops in Case~1, 
then $Q_i=\emptyset$, and it follows 
that every state  in $S_*$ can reach $T_*$.
(ii)~We now consider the case when the algorithm stops in Case~2: 
In this case every state in $J_i$ has a path to $T_i=T_*$, this is because if there 
is a state $s$ in $J_i$ with no path to $T_i$, then the DFS tree from $s$ would have been
identified as $Q_i$ in step 2 (a) and the algorithm would not have stopped.
It follows that there is no bottom scc in the graph induced by $S_*$ that does not 
intersect $T_*$: because if there is a bottom scc that does not contain 
a state from $J_i$ and also does not contain a target state, then it would have 
been identified in the last iteration of Case~1.
Since every state in $S_*$ has an out-going edge, 
it follows every state in $S_*$ has a path to $T_*$. 
Hence (b) follows.
Consider a shortest path (or the BFS tree) from all states in $S_*$ to 
$T_*$, and for a state $s \in S_* \cap \SA$, let $s'$ be the successor for 
the shortest path, and we consider the pure memoryless strategy $\straa_*$ 
that chooses the shortest path successor for all states 
$s \in (S_* \setminus T_*) \cap \SA$, and in states in $T_* \cap \SA$ choose
any successor in $S_*$.
Let $\ell=|S_*|$ and let $\alpha$ be the minimum of the positive transition 
probability of the MDP.
For all states $s \in S_*$, the probability that $T_*$ is reached within $\ell$ 
steps is at least $\alpha^\ell$, and it follows that the probability that $T_*$
is not reached within $k \times \ell$ steps is at most $(1-\alpha^\ell)^k$, 
and this goes to~0 as $k$ goes to $\infty$. 
It follows that for all $s \in S_*$ the pure memoryless strategy $\straa_*$ 
ensures that $T_*$ is reached with probability~1.
Moreover, the strategy ensures that $S_*$ is never left, and hence it follows
that $T_*$ is visited infinitely often with probability~1.
It follows that $S_* \subseteq \waa(\Buchi(T_*)) \subseteq \waa(\Buchi(T))$ and 
hence the correctness follows.
\qed
\end{proof}

We now analyze the running time of the algorithm. 
\begin{lemma}
Given an MDP $G$ with $m$ edges, Algorithm \impa\ takes $O(m\cdot \sqrt{m})$ time.
\end{lemma}
\begin{proof}
The total work of the algorithm, when Case~1 is executed, over all 
iterations is at most $O(\sqrt{m} \cdot m)$: this follows because between 
two iterations of Case~1 at least $\sqrt{m}$ edges must have been removed
from the graph (since $|J_i| > \sqrt{m}$ everytime Case~1 is executed 
other than the case when $i=0$),  and hence Case~1 can be executed at most 
$m/\sqrt{m}=\sqrt{m}$ times.
Since each iteration can be achieved in $O(m)$ time, the $O(m\cdot\sqrt{m})$ 
bound for Case~1 follows.
We now show that the total work of the algorithm, when Case~2 is executed, 
over all iterations is at most $O(\sqrt{m} \cdot m)$. The argument is as 
follows: consider an iteration $i$ such that Case~2 is executed.
Then we have $|J_i| \leq \sqrt{m}$.
Let $Q_i$ be the DFS tree in iteration $i$ while executing 
Case~2, and let $E(Q_i)= \cup_{s \in Q_i} E(s)$. The lock-step search ensures 
that the number of edges explored in this iteration is at most 
$|J_i| \cdot |E(Q_i)| \leq \sqrt{m} \times |E(Q_i)|$. 
Since $Q_i$ is removed from the graph we \emph{charge} the work of $\sqrt{m} \cdot |E(Q_i)|$ 
to edges in $E(Q_i)$, charging work $\sqrt{m}$ to each edge. 
Since there are at most $m$ edges, the total charge of the work over all 
iterations when Case~2 is executed is at most $O(m \cdot \sqrt{m})$. 
Note that if instead of $\sqrt{m}$ we would have used a bound $k$ in distinguishing 
Case~1 and Case~2, we would have achieved a running time bound of 
$O(m^2/k + m\cdot k)$, which is optimized by $k=\sqrt{m}$.
Our desired result follows.
\qed
\end{proof}
This gives us the following result.

\begin{theorem}\label{thrm_as_reach}
Given an MDP $G$ and a set $T$ of target states, the algorithm \impa\  
correctly computes the set $\waa(\Buchi(T))$ in time $O(m \cdot \sqrt{m})$.
\end{theorem}

\subsection{Symbolic implementation of \impa}
\label{sec:symimpa}
In this subsection we will a present symbolic implementation of each of the steps of algorithm 
\impa. 
The symbolic algorithm depends on the following symbolic operations that can 
be easily achieved with an \obdd\ implementation.
For a set $X \subseteq S$ of states, let
\[
\begin{array}{rcl}
\Pre(X) & = & \set{s \in S \mid E(s) \cap X \neq \emptyset}; \\ %\qquad 
\Post(X) & = &  \set{t \in S \mid  t \in \bigcup_{s \in  X} E(s)}; \\
\CPre(X) & = & \set{s \in \SR \mid E(s) \cap X \neq \emptyset} \cup \set{s \in \SA \mid E(s) \subseteq X}.
\end{array}
\]
In other words, $\Pre(X)$ is the predecessors of states in $X$;
$\Post(X)$ is the successors of states in $X$; and 
$\CPre(X)$ is the set of states $Y$ such that for every random state in $Y$ 
there is a successor in $X$, and for every player 1 state in $Y$ all successors
are in $Y$.

We now present a symbolic version of \impa. 
For the symbolic version the basic steps are as follows: 
(i)~Case 1 of the algorithm is same as Case~1 of \impa, and 
(ii)~Case 2 is similar to Case~2 of \impa, and the only change in Case 2 is 
instead of lock-step search exploring the same number of edges, we have lock-step search
that executes the same number of symbolic steps. 
The details of the symbolic implementation are as follows,
and we will refer to the algorithm as \simpa.
\begin{enumerate}

\item \emph{Case 1.} In Case 1(a) we need to compute reachability 
to a target set $T$. 
The symbolic implementation is standard and done as follows:
$X_0=T$ and $X_{i+1} := X_i \cup \Pre(X_i)$ until $X_{i+1} = X_i$.
The computation of the random attractor is also standard and is achieved as above 
replacing $\Pre$ by $\CPre$.
%%All the other operations are constant symbolic operations. 
It follows that every iteration of Case 1 can be achieved in $O(n)$ symbolic 
steps.

\item \emph{Case 2.} For analysis of Case 2 we present a symbolic 
implementation of the lock-step forward search. 
The lock-step ensures that each search executes the same number of 
symbolic steps.
The implementation of the forward search from a state $s$ in iteration $i$ is 
achieved as follows: 
$P_0:=\set{s}$ and $P_{j+1}:= P_j \cup \Post(P_j)$ unless $P_{j+1}=P_j$ or 
$P_j \cap T_i\neq \emptyset$. 
If $P_{j} \cap T_i \neq \emptyset$, then the forward search is stopped from $s$.
If $P_{j+1} =P_j$ and $P_j \cap T_i=\emptyset$, then we have identified that 
there is no path from states in $P_j$ to $T_i$. 
 
\item \emph{Symbolic computation of cardinality of sets.} 
The other key operation required by the algorithm is determining whether 
the  size of set $J_i$ is at least $\sqrt{m}$ or not. 
%%Also observe that the set $J_i$ is obtained symbolically as the set
%%$\Pre(L_i) \cup L_i$. 
Below we describe the details of this symbolic operation.
%%and the symbolic steps requirement  analysis.

\end{enumerate}

\noindent{\bf Symbolic computation of cardinality.} 
Given a symbolic description of a set $X$ and a number $k$, 
our goal is to determine whether $|X| \leq k$. 
A naive way is to check for each state, whether it belongs to $X$.  
But this takes time proportional to the size of state space and 
also is not symbolic.
We require a procedure that uses the structure of a BDD and directly finds the states which this BDD represents.
It should also take into account that if more than $k$ states are already 
found, then no more computation is required.
We present the following procedure to accomplish the same. 
A \emph{cube} of a BDD is a path from root node to leaf node where the leaf node is the constant 1 (i.e. true).
Thus, each cube represents a set of states present in the BDD which are exactly the states found by doing every possible assignment of the variables not occurring in the cube.
For an explicit implementation: consider a procedure that uses Cudd\_ForEachCube 
(from CUDD package, see~\cite{CuDD} for symbolic implementation) 
to iterate over the cubes of a given \obdd\ in the same manner the successor function works on a binary tree.
If $l$ is the number of variables not occurring in a particular cube, we get $2^l$ states from that cube which are part of the \obdd. 
We keep on summing up all such states until they exceed $k$. 
If it does exceed, we stop and say that $|X| > k$. Else we terminate when we have exhausted all cubes and we get $|X| \leq k$.
Thus we require $\min (k, |\mathit{BDD}(X)|)$ symbolic steps, where 
$\mathit{BDD}(X)$ is the size of the \obdd\ of $X$. 
% The running time of this procedure is min(k*log(n), |BDD(X)|) because each call to next_cube takes O(log n) time and we cannot take more than |BDD(X)| even if we iterate through all cubes.
%We also maintain an array of state which we use in case $|S|<n$. 
%We fill this array by obtaining the states represented by this cube. 
%These states can be obtained by assigning each possible value to the unassigned variables in the cube. \\
We also note, that this method operates on \obdd s that represent 
set of states, and these \obdd s only use $\log(n)$ variables compared to 
$2 \cdot \log(n)$ variables used by \obdd s representing transitions 
(edge relation). 
Hence, the operations mentioned are cheaper as compared to  
$\Pre$ and $\Post$ computations.

\smallskip\noindent{\bf Correctness and runtime analysis.} 
The correctness of \simpa\ is established following the correctness arguments 
for algorithm \impa. 
We now analyze the worst case number of symbolic steps.
The total number of symbolic steps executed by Case 1 over all iterations is 
$O(n \cdot \sqrt{m})$ since between two executions of Case~1 at least 
$\sqrt{m}$ edges are removed, and every execution is achieved in $O(n)$ 
symbolic steps.
The work done for the symbolic cardinality computation is charged to the edges 
already removed from the graph, and hence the total number of symbolic steps 
over all iterations for the size computations is $O(m)$. 
We now show that the total number of symbolic steps executed over all 
iterations of Case 2 is $O(n \cdot \sqrt{m})$.
The analysis is achieved as follows. 
Consider an iteration $i$ of Case~2,  and let the number of states removed 
in the iteration be $n_i$. 
Then the number of symbolic steps executed in this iteration for each of 
the forward search is at most $n_i$, and since $|J_i| \leq \sqrt{m}$, it follows 
that the number of symbolic steps executed is at most $n_i \cdot \sqrt{m}$.
Since we remove $n_i$ states, we \emph{charge} each state removed from the graph 
with $\sqrt{m}$ symbolic steps for the total $n_i \cdot \sqrt{m}$ symbolic steps. 
Since there are at most $n$ states, the total charge of symbolic steps over all 
iterations is $O(n \cdot \sqrt{m})$.
Thus it follows that we have a symbolic algorithm to compute the almost-sure
winning set for MDPs with B\"uchi objectives in 
$O(n \cdot \sqrt{m})$ symbolic steps.

\begin{theorem}\label{thrm_as_reach_symbolic}
Given an MDP $G$ and a set $T$ of target states, the symbolic algorithm 
\simpa\  correctly computes $\waa(\Buchi(T))$ in 
$O(n \cdot \sqrt{m})$ symbolic steps.
\end{theorem}

\begin{remark}{}
In many practical cases, MDPs have constant out-degree and hence we obtain a
symbolic algorithm that works in $O(n \cdot \sqrt{n})$ symbolic steps, as 
compared to the previous known (symbolic implementation of
the classical) algorithm that requires $\Omega(n^2)$ symbolic steps.
\end{remark}

%%%%% KRISH_J1: CHECK START
\begin{remark}{}
Note that in our algorithm we used $\sqrt{m}$ to distinguish between Case~1 and Case~2
to obtain the optimal time complexity. However, our algorithm could also be parametrized
with a parameter $k$ to distinguish between Case~1 and Case~2, and then the number of 
symbolic steps required is $O( \frac{n \cdot m}{k} + n \cdot k)$.
For example, if $m=O(n)$, by choosing $k=\log n$, we obtain a symbolic algorithm 
that requires $O(\frac{n^2}{\log n})$ symbolic steps as compared to the $O(n^2)$ 
symbolic steps of the previous known algorithms.
In other words our algorithm can be easily parametrized to provide a trade-off between 
the number of forward searches and speed up in the number of symbolic steps.
\end{remark}
%%%%% KRISH_J1: CHECK START

\subsection{Optimized \simpa}\label{subsec:simpa}
In the worst case, the \simpa\ algorithm takes $O(n \cdot \sqrt{m})$ steps.
However it is easy to construct a family of MDPs with $n$ states and 
$O(n)$ edges, where the classical algorithm takes $O(n)$ symbolic steps, 
whereas \simpa\ requires $\Omega(n \cdot \sqrt{n})$ symbolic steps. 
One approach to obtain an algorithm that takes at most $O(n \cdot \sqrt{n})$ symbolic 
steps and no more than linearly many symbolic steps of the 
classical algorithm is to dovetail (or run in lock-step) the classical algorithm 
and \simpa, and stop when either of them stops. 
This approach will take time at least twice the minimum running time of the 
classical algorithm and \simpa. 
We show that a much smarter dovetailing is possible (at the level of each 
iteration).
We now present the smart dovetailing algorithm, and we call the
algorithm  \oimpa. 
The basic change is in Case 2 of \simpa. We now describe the changes in 
Case 2:
\begin{itemize}

\item At the beginning of an execution of Case 2 at iteration $i$ such that 
the last execution was Case 1, we initialize a set $U_i$ to $T_i$. 
Every time a post computation ($\Post(P_j)$) is done, we update $U_i$ by 
$U_{i+1}:= U_i \cup \Pre(U_i)$ (this is the backward exploration step of the classical algorithm
and it is dovetailed with the forward exploration step in every iteration). 
For the forward exploration step, we continue the computation of $P_j$ unless
$P_{j+1}=P_j$ or $P_{j} \cap U_i \neq \emptyset$ 
(i.e., \simpa\ checked the emptiness of intersection with $T_i$, whereas 
in \oimpa\ the emptiness of the intersection is checked with $U_i$).
If $U_{i+1}=U_i$ (i.e., a fixpoint is reached), then $S_i \setminus U_i$ 
and its random attractor is removed from the graph.
\end{itemize}

\noindent{\bf Correctness and symbolic steps analysis.} 
We present the correctness and number of symbolic steps required analysis 
for the algorithm \oimpa.
The correctness analysis is same as \impa\ and the only change is as follows 
(we describe iteration $i$): 
(a)~if in Case 2 we obtain a set $P_j=P_{j+1}$ and its intersection with $U_i$ is 
empty, then there is no path from $P_j$ to $U_i$ and since $T_i \subseteq U_i$, 
it follows that there is no path from $P_j$ to $U_i$; 
(b)~if $P_j \cap U_i \neq \emptyset$, then since $U_i$ is obtained as the backward
exploration from $T_i$, every state in $U_i$ has a path to $T_i$, and it follows that 
there is a path from the starting state of $P_j$ to $U_i$ and hence to $T_i$;
and 
(c)~if $U_i=\Pre(U_i)$, then $U_i$ is the set of states that can reach $T_i$ and 
all the other states can be removed.  
Thus the correctness follows similar to the arguments for \impa. 
The key idea of the running time analysis is as follows:
\begin{enumerate}
\item Case 1 of the algorithm is same  to Case 1 of \simpa, and in Case 2 the 
algorithm also runs like \simpa, but for every symbolic step ($\Post$ 
computation) of \simpa, there is an additional ($\Pre$) computation. 
Hence the total number of symbolic steps of \oimpa\ is at most twice the 
number of symbolic steps of \simpa.
However, the optimized step of maintaining the set $U_i$ which includes $T_i$ may allow
to stop several of the forward exploration as they may intersect with $U_i$ earlier
than intersection with $T_i$.

\item Case 1 of the algorithm is same as in Case 1 of the classical algorithm. 
In Case 2 of the algorithm the backward exploration step is the same as the classical algorithm, 
and (i)~for every $\Pre$ computation, there is an additional $\Post$ computation and 
(ii)~for every check whether $U_i=\Pre(U_i)$, there is a check whether 
$P_j = P_{j+1}$ or $P_j \cap U_i \neq \emptyset$. 
It follows that the total number of symbolic steps of Case 1 and Case 2 over 
all iterations is at most twice the number of symbolic steps of the classical algorithm.
The cardinality computation takes additional $O(m)$ symbolic steps over all iterations.
\end{enumerate}
Hence we obtain the following result.
%%Hence it follows that \oimpa\ takes at most 
%%\[
%%\min\set{2\cdot\SymbStep(\mbox{\simpa}), 2\cdot\SymbStep(\mbox{\classical}) + O(m)}
%%\] 
%%symbolic steps.
%%Details are given in appendix and we have the following result.

\begin{theorem}\label{thrm_as_reach_opt_symbolic}
Given an MDP $G$ and a set $T$ of target states, the symbolic algorithm 
\oimpa\  correctly computes $\waa(\Buchi(T))$ and requires at most 
\[
\min\set{2\cdot\SymbStep(\mbox{\simpa}), 2\cdot\SymbStep(\mbox{\classical}) + O(m)}
\] 
symbolic steps, where $\SymbStep$ is the number of symbolic steps of an
algorithm.
\end{theorem}

Observe that it is possible that the number of symbolic steps and running time 
of \oimpa\ is smaller than both \simpa\ and \classical\ (in contrast to a 
simple dovetailing of \simpa\ and \classical,  where the running time and 
symbolic steps is twice that of the minimum).
It is straightforward to construct a family of examples where \oimpa\ 
takes linear ($O(n)$) symbolic steps, however both \classical\ and \simpa\ take at least
$O(n \cdot \sqrt{n})$ symbolic steps.
% and our experimental results 
%shows indeed for several cases \oimpa\ outperforms both the \simpa\ and
%the \classical\ algorithm. {\bf KRISH: check the last sentence with Manas and Nisarg.}

%%{\bf KRISH: change the whole section to B\"uchi. To be done later. $T_i, T_*$ etc.}

\newcommand{\wl}{{\sc WinLose}}
\newcommand{\iwl}{{\sc ImprWinLose}}
\newcommand{\siwl}{{\sc SymbImprWinLose}}

\section{The Win-Lose Algorithm}\label{sec:winlose}
All the algorithms known for computing the almost-sure winning set (including 
the algorithms presented in the previous section) iteratively compute 
the set of states from where it is guaranteed that there is no almost-sure 
winning strategy for the player. 
The almost-sure winning set is discovered only when the algorithm stops. 
In this section, first we will present an algorithm that iteratively 
computes two sets $W_1$ and $W_2$, where $W_1$ is a subset of the 
almost-sure winning set, and $W_2$ is a subset of the complement of the 
almost-sure winning set. 
The algorithm has $O(K \cdot m)$ running time, where $K$ is the size of the
maximal strongly connected component (scc) of the graph of the MDP. 
We will first present the basic version of the algorithm, and
then present an improved version of the algorithm, using the techniques to
obtain \impa\ from the classical algorithm, and finally present the 
symbolic implementation of the new algorithm.

\subsection{The basic win-lose algorithm} The basic steps of the new 
algorithm %, that we call \emph{win-lose} algorithm, 
are as follows.
The algorithm maintains $W_1$ and $W_2$, that are guaranteed to be subsets 
of the almost-sure winning set and its complement respectively. 
Initially $W_1=\emptyset$ and $W_2=\emptyset$. 
We also maintain that $W_1=\attr_{1}(W_1)$ and $W_2=\attr_{R}(W_2)$.
We denote by $W$ the union of $W_1$ and $W_2$.
We describe an iteration of the algorithm and we will refer to the algorithm 
as the \wl\ algorithm (pseudocode is given as Algorithm~\ref{algorithm:WinLose}).
\begin{enumerate}
\item \emph{Step 1.} Compute the scc decomposition of 
the remaining graph of the MDP, i.e., scc decomposition of the MDP graph 
induced by $S\setminus W$. 
\item \emph{Step 2.} For every bottom scc $C$ in the remaining graph: 
if $C \cap \Pre(W_1) \neq \emptyset$ or $C \cap T \neq \emptyset$, then 
$W_1 =\attr_1(W_1 \cup C)$; 
else $W_2 =\attr_{R}(W_2 \cup C)$, and the states in $W_1$ and $W_2$ are 
removed from the graph.
\end{enumerate}
The stopping criterion is as follows: the algorithm stops when $W=S$.
Observe that in each iteration, a set $C$ of states is included in 
either $W_1$ or $W_2$, and hence $W$ grows in each iteration. 
Observe that our algorithm has the flavor of iterative scc decomposition 
algorithm for computing maximal end-component decomposition of MDPs.

\begin{algorithm}[t]
\caption{\bf WinLose}
\label{algorithm:WinLose}
{
\begin{tabbing}
aa \= aa \= aa \= aa \= aa \= aa \= aa \= aa \= aa \= aa \= aa \= aa \kill
\> \textbf{Input}: An MDP $G=((S, E), (\SA,\SR),\trans)$ with B\"uchi set $T$. \\
\> \textbf{Output}: $\waa(\Buchi(T))$, i.e., the almost-sure winning set for player~1. \\
\> 1. $W := W_{1} := W_{2} := \emptyset$ \\
\> 2. \textbf{while}($W \neq S$) \textbf{do} \\
\> \> 2.1. $\mathit{SCCS} := $\textbf{SCC-Decomposition}$(S \setminus W)$ 
(i.e. scc decomposition of the graph induced by $S\setminus W$)\\
\> \> 2.2. \textbf{for each} $C$ in $\mathit{SCCS}$ \\
\> \> \> 2.2.1. \textbf{if} $(E(C) \subset C \cup W)$ \textbf{then} (checks 
if $C$ is a bottom scc in graph induced by $S\setminus W$)\\
\> \> \> \> 2.2.1.1. \textbf{if} $C \cap T \neq \emptyset$ or $E(C) \cap W_{1} \neq \emptyset$ \textbf{then} \\
\> \> \> \> \> 2.2.1.1.1. $W_{1} := W_{1} \cup C$ \\
\> \> \> \> 2.2.1.2. \textbf{else} \\
\> \> \> \> \> 2.2.1.2.1. $W_{2} := W_{2} \cup C$ \\
\> \> 2.3. $W_{1} := \attr_{1}(W_{1}, (S, E),(\SA,\SR))$ \\
\> \> 2.4. $W_{2} := \attr_{R}(W_{2}, (S, E),(\SA,\SR))$ \\
\> \> 2.5. $W := W_{1} \cup W_{2}$ \\
\> 3. \textbf{return} $W_{1}$ 
\end{tabbing}
}
\end{algorithm}

\noindent{\bf Correctness of the algorithm.} Note that in Step~2 we ensure that 
$\attr_1(W_1)=W_1$ and $\attr_R(W_2)=W_2$, and hence in the remaining graph 
there is no state of player~1 with an edge to $W_1$ and no random state with
an edge to $W_2$. 
We show by induction that after every iteration 
$W_1 \subseteq \waa(\Buchi(T))$ and $W_2 \subseteq S\setminus \waa(\Buchi(T))$. 
The base case (with $W_1=W_2=\emptyset$) follows trivially.
We prove the inductive case considering the following two cases.
\begin{enumerate}

\item Consider a bottom scc $C$ in the remaining graph such that 
$C \cap \Pre(W_1) \neq \emptyset$ or $C \cap T\neq \emptyset$. 
Consider the randomized memoryless strategy $\straa$ for the player that 
plays all edges in $C$ uniformly at random, i.e., for $s \in C$ we have 
$\straa(s)(t)=\frac{1}{|E(s) \cap C|}$ for $t \in E(s) \cap C$. 
If $C \cap \Pre(W_1) \neq \emptyset$, then the strategy ensures that $W_1$ 
is reached with probability~1, since $W_1 \subseteq \waa(\Buchi(T))$ by 
inductive hypothesis it follows $C \subseteq \waa(\Buchi(T))$.
Hence $\attr_1(W_1 \cup C) \subseteq \waa(\Buchi(T))$.
If $C \cap T \neq \emptyset$, then since there is no edge from random states 
to $W_2$, it follows that under the randomized memoryless strategy $\straa$,
the set $C$ is a closed recurrent set of the resulting Markov chain, and 
hence every state is visited infinitely often with probability~1.
Since $C \cap T \neq \emptyset$, it follows that $C \subseteq \waa(\Buchi(T))$,
and hence $\attr_1(W_1 \cup C) \subseteq \waa(\Buchi(T))$.

\item Consider a bottom scc $C$ in the remaining graph such that 
$C \cap \Pre(W_1)=\emptyset$ and $C \cap T=\emptyset$. Then consider any 
strategy for player~1: 
(a)~If a play starting from a state in $C$ stays in the remaining graph, 
then since $C$ is a bottom scc, it follows that the play stays in $C$ 
with probability~1. Since $C \cap T =\emptyset$ it follows that $T$ is 
never visited. 
(b)~If a play leaves $C$ (note that $C$ is a bottom scc of the remaining 
graph and not the original graph, and hence a play may leave $C$), 
then since $C \cap \Pre(W_1)=\emptyset$, it follows 
that the play reaches $W_2$, and by hypothesis $W_2 \subseteq S\setminus 
\waa(\Buchi(T))$. 
In either case it follows that $C \subseteq S \setminus \waa(\Buchi(T))$.
It follows that $\attr_{R}(W_2 \cup C) \subseteq S\setminus \waa(\Buchi(T))$.
\end{enumerate}
The correctness of the algorithm follows as when the algorithm stops 
we have $W_1 \cup W_2 =S$.

\noindent{\bf Running time analysis.} In each iteration of 
the algorithm at least one state is removed from the graph, and every 
iteration takes at most $O(m)$ time: in every iteration, the scc 
decomposition of step~1 and the attractor computation in step~2 can be 
achieved in $O(m)$ time.
Hence the naive running of the algorithm is $O(n \cdot m)$. 
The desired $O(K\cdot m)$ bound is achieved by considering the standard
technique of running the algorithm on the scc decomposition of the MDP.
In other words, 
%%A simple modification of the algorithm is as follows: 
we first compute the scc of the graph of the MDP, and then proceed bottom 
up computing the partition $W_1$ and $W_2$ for an scc $C$ once the partition is 
computed for all states below the scc. 
Observe that the above correctness arguments are still valid.
The running time analysis is as follows:
let $\ell$ be the number of scc's of the graph, and let 
$n_i$ and $m_i$ be the number of states and edges of the $i$-th scc. 
Let $K=\max \set{n_i \mid 1\leq i \leq \ell}$.
Our algorithm runs in time $O(m) + \sum_{i=1}^\ell O(n_i \cdot m_i) 
\leq O(m) + \sum_{i=1}^\ell O(K \cdot m_i) = O(K \cdot m)$.

\begin{theorem}
Given an MDP with a B\"uchi objective, the \wl\ algorithm iteratively
computes the subsets of the almost-sure winning set and its complement, 
and in the end correctly computes the set $\waa(\Buchi(T))$ and 
the algorithm runs in time $O(K_S \cdot m)$, where $K_S$ is the maximum 
number of states in an scc of the graph of the MDP.
\end{theorem}

\subsection{Improved \wl\ algorithm and symbolic implementation}

\noindent{\bf Improved \wl\ algorithm.}
The improved version of the \wl\ algorithm performs a forward exploration 
to obtain a bottom scc like Case 2 of \impa.
At iteration $i$, we denote the remaining subgraph as $(S_i,E_i)$, where 
$S_i$ is the set of remaining states, and $E_i$ is the set of remaining 
edges.
The set of states removed will be denoted by $Z_i$, i.e., $S_i=S \setminus Z_i$,
and $Z_i$ is the union of $W_1$ and $W_2$.
In every iteration the algorithm identifies a set $C_i$ of states such that
$C_i$ is a bottom scc in the remaining graph, and then it follows the 
steps of the \wl\ algorithm.
We will consider two cases.
The algorithm maintains %a set $I_{i+1}$ as the union of 
the set $L_{i+1}$ of states that were removed from the graph since (and including) 
the last iteration of Case~1, and the set $J_{i+1}$ of states that lost an edge 
to states removed from the graph since the last iteration of Case~1. 
Initially $J_0:=L_0:=Z_0:=W_1:=W_2:=\emptyset$, and let $i:=0$, and we describe the 
iteration $i$ of our algorithm. 
We call our algorithm \iwl\ 
(pseudocode is given as Algorithm~\ref{algorithm:ImprWinLose}).

\begin{enumerate}
\item \emph{Case~1.} If ($(|J_i| > \sqrt{m})$ or $i=0$), then 
\begin{enumerate}
\item Compute the scc decomposition of the remaining graph. 
%%%and let $C_i$ be a bottom scc.

\item For each bottom scc $C_i$, if $C_i \cap T \neq \emptyset$ or $C_i \cap \Pre(W_1) \neq \emptyset$, 
then $W_1:= \attr_1(W_1 \cup C_i)$, else $W_2:= \attr_{R}(W_2 \cup C_i)$.

\item $Z_{i+1}:= W_1 \cup W_2$. The set $Z_{i+1} \setminus Z_i$ 
is removed from the graph.

\item The set $L_{i+1}$ is the set of states removed from the graph in
this iteration and $J_{i+1}$ 
be the set of states in the remaining graph with an edge to $L_{i+1}$. 

\item If $Z_i$ is $S$, the algorithm stops, 
otherwise $i:=i+1$ and go to the next iteration.

\end{enumerate}

\item \emph{Case 2.} Else  ($|J_i| \leq \sqrt{m})$ and $i>0$), then 
\begin{enumerate}
\item Consider the set $J_{i}$ to be the set of vertices in the graph that 
lost an edge to the states removed since the last iteration that 
executed Case~1.   

\item We do a lock-step search from every state $s$ in $J_i$ as follows: 
we do a DFS from $s$, until the DFS stops. 
Once the DFS stops we have identified a bottom scc $C_i$. 

\item If $C_i \cap T \neq \emptyset$ or $C_i \cap \Pre(W_1) \neq \emptyset$, 
then $W_1:= \attr_1(W_1 \cup C_i)$, else $W_2:= \attr_{R}(W_2 \cup C_i)$.

\item $Z_{i+1}:= W_1 \cup W_2$. The set $Z_{i+1} \setminus Z_i$ 
is removed from the graph.

\item The set $L_{i+1}$ is the set of states removed from the graph since 
the last iteration of Case~1 and $J_{i+1}$ be the set of states in the 
remaining graph with an edge to $L_{i+1}$.

\item If $Z_i=S$, the algorithm stops, otherwise $i:=i+1$ and go to the next iteration.

\end{enumerate}

\end{enumerate}

\begin{algorithm}[!ht]
\caption{\bf ImprWinLose}
\label{algorithm:ImprWinLose}
{
\begin{tabbing}
aa \= aa \= aa \= aa \= aa \= aa \= aa \= aa \= aa \= aa \= aa \= aa \kill
\> \textbf{Input}: An MDP $G=((S, E), (\SA,\SR),\trans)$ with B\"uchi set $T$. \\
\> \textbf{Output}: $\waa(\Buchi(T))$, i.e., the almost-sure winning set for player~1. \\
\> 1. $i := 0$; $S_{0} := S$; $E_{0} := E$; $T_{0} := T$; \\
\> 2. $W_{1} := W_{2} := L_{0} := Z_{0} := J_{0} := \emptyset $; \\
\> 3. \textbf{if} $(|J_{i}| > \sqrt{m}$ or $i = 0)$ \textbf{then}\\
\> \> 3.1. $\mathit{SCCS} := $\textbf{SCC-Decomposition}$(S_{i})$ (scc decomposition of graph induced by $S_i$)\\
\> \> 3.2. \textbf{for each} $C$ in $\mathit{SCCS}$ \\
\> \> \> 3.2.1. \textbf{if} $(E_{i}(C) \subset C)$ \textbf{then} (checks if $C$ is a bottom scc in graph induced by $S_i$)\\
\> \> \> \> 3.2.1.1. \textbf{if} $C \cap T \neq \emptyset$ or $E(C) \cap W_{1} \neq \emptyset$ \textbf{then} \\
\> \> \> \> \> 3.2.1.1.1. $W_{1} := W_{1} \cup C$ \\
\> \> \> \> 3.2.1.2. \textbf{else} \\
\> \> \> \> \> 3.2.1.2.1. $W_{2} := W_{2} \cup C$ \\
\> \> 3.3. \textbf{goto} line $5$ \\
\> 4. \textbf{else} (i.e., $J_{i} \leq \sqrt{m}$ and $i>0$) \\
\> \> 4.1. \textbf{for each} $s \in J_{i}$ \\
\> \> \> 4.1.1. $\mathit{DFS}_{i,s} := s$ (initializing DFS-trees)\\
\> \> 4.2. \textbf{for each} $s \in J_{i}$ \\
\> \> \> 4.2.1. Do 1 step of DFS from $\mathit{DFS}_{i,s}$\\
\> \> \> 4.2.2. \textbf{if} DFS completes \textbf{then} \\
\> \> \> \> 4.2.2.1. $C := \mathit{DFS}_{i,s}$\\
\> \> \> \> 4.2.2.2. \textbf{if} $C \cap T \neq \emptyset$ or $E(C) \cap W_{1} \neq \emptyset$ \textbf{then} \\
\> \> \> \> \> 4.2.2.2.1. $W_{1} := W_{1} \cup C$ \\
\> \> \> \> 4.2.2.3. \textbf{else} \\
\> \> \> \> \> 4.2.2.3.1. $W_{2} := W_{2} \cup C$ \\
\> \> \> \> 4.2.2.4. \textbf{goto} line $5$ \\
\> 5. Removal of $W_{1}$ and $W_{2}$ states in the following steps \\
\> \> 5.1. $W_{1} := \attr_{1}(W_{1}, (S_{i}, E_{i}),(\SA\cap S_i,\SR\cap S_i))$ \\
\> \> 5.2. $W_{2} := \attr_{R}(W_{2}, (S_{i}, E_{i}),(\SA\cap S_i,\SR\cap S_i))$ \\
\> \> 5.3. $Z_{i+1} := Z_{i} \cup W_{1} \cup W_{2}$ \\
\> \> 5.4. $S_{i+1} := S_{i} \setminus Z_{i+1}$; $E_{i+1} := E_{i} \cap S_{i+1} \times S_{i+1}$ \\
\> \> 5.5. \textbf{if} the last \textbf{goto} call was from line $3.3$ \textbf{then} \\
\> \> \> 5.5.1. $L_{i+1} := Z_{i+1} \setminus Z_{i}$ \\
\> \> 5.6. \textbf{else} \\
\> \> \> 5.6.1. $L_{i+1} := L_{i} \cup (Z_{i+1} \setminus Z_{i})$ \\
\> \> 5.7. $J_{i+1} := E^{-1}(L_{i+1}) \cap S_{i+1}$ \\
\> \> 5.8. \textbf{if} $Z_{i+1} = S$ \textbf{then} \\
\> \> \> 5.8.1. goto line $6$ \\
\> \> 5.9. $i := i + 1$; \textbf{goto} line $3$ \\
\> 6. \textbf{return} $W_{1}$ \\
\end{tabbing}
}
\end{algorithm}

\medskip\noindent{\bf Correctness and running time.} 
The correctness proof of \iwl\ is similar as the correctness argument 
of \wl\ algorithm.
One additional care requires to be taken for Case~2: we need to show 
that when we terminate the lockstep DFS search in Case~2, then 
we obtain a bottom scc.
First, we observe that in iteration $i$, when Case~2 is executed, 
each bottom scc must contain a state from $J_i$, since it was not 
a bottom scc in the last execution of Case 1.
Second, among all the lockstep DFSs, the first one that terminates 
must be a bottom scc because the DFS search from a state of $J_i$ 
that does not belong to a bottom scc explores states of bottom scc's 
below it.
Since Case~2 stops when the first DFS terminates we obtain a bottom scc.
The rest of the correctness proof is identical as the proof for 
the \wl\ algorithm.
The running time analysis of the algorithm is similar to \impa\ algorithm, and this
shows the algorithm runs in $O(m \cdot \sqrt{m})$ time. 
Applying the \iwl\ algorithm bottom up on the scc decomposition of the MDP gives us 
a running time of $O(m \cdot \sqrt{K_E})$, where $K_E$ is the maximum number of 
edges of an scc of the MDP.

\begin{theorem}
Given an MDP with a B\"uchi objective, the \iwl\ algorithm iteratively
computes the subsets of the almost-sure winning set and its complement, 
and in the end correctly computes the set $\waa(\Buchi(T))$.
The algorithm \iwl\ runs in time $O(\sqrt{K_E} \cdot m)$, where $K_E$ is the maximum 
number of edges in an scc of the graph of the MDP.
\end{theorem}

\smallskip\noindent{\bf Symbolic implementation.} The symbolic implementation of 
\iwl\ algorithm is obtained in a similar fashion as \simpa\ was obtained from 
\impa. The only additional step required is the symbolic scc computation. It follows
from the results of~\cite{GPP} that scc decomposition can be computed in 
$O(n)$ symbolic steps. 
In the following section we will present an improved symbolic scc computation 
algorithm. 
The correctness proof of \siwl\ is similar to \iwl\ algorithm.
For the correctness of the \siwl\ algorithm we again need to 
take care that when we terminate in Case~2, then we have identified a
bottom scc. 
Note that for symbolic step forward search we cannot guarantee that 
the forward search that stops first gives a bottom scc.
For Case~2 of the \siwl\ we do in lockstep both symbolic forward and 
backward searches, stop when both the searches stop and gives
the same result. 
Thus we ensure when we terminate an iteration of Case~2 we obtain a
bottom scc. 
The correctness then follows from the correctness arguments of \wl\ 
and \iwl.
The symbolic steps required analysis is same as for \simpa.

\begin{corollary}
Given an MDP with a B\"uchi objective, the symbolic \iwl\ algorithm (\siwl) 
iteratively computes the subsets of the almost-sure winning set and its 
complement, and in the end correctly computes the set $\waa(\Buchi(T))$. 
The algorithm \siwl\  requires $O(\sqrt{K_E} \cdot n)$ symbolic steps, where $K_E$ is the maximum 
number of edges in an scc of the graph of the MDP.
\end{corollary}

\begin{remark}
It is clear from the complexity of the \wl\ and \iwl\ algorithms that they 
would perform better for MDPs where the graph has many small scc's, rather 
than few large ones.
\end{remark}

\newcommand{\FW}{\mathsf{FW}}
\newcommand{\skelfwd}{{\sc SkelFwd}}
\newcommand{\iskelfwd}{{\sc ImprovedSkelFwd}}
\newcommand{\sym}{{\sc SymbolicScc}}
\newcommand{\isym}{{\sc ImprovedSymbolicScc}}
\newcommand{\sccfind}{{\sc SCCFind}}
\newcommand{\isccfind}{{\sc ImprovedSCCFind}}
\newcommand{\newstate}{\mathsf{NewState}}
\newcommand{\newset}{\mathsf{NewSet}}
\newcommand{\scc}{\mathsf{SCC}}
\newcommand{\fw}{\mathsf{FWSet}}

\section{Improved Symbolic SCC Algorithm}\label{sec:symscc}
A symbolic algorithm to compute the scc decomposition of a graph in 
$O(n \cdot \log n)$ symbolic steps was presented in~\cite{BGS00}.
The algorithm of~\cite{BGS00} was based on forward and backward searches.
The algorithm of~\cite{GPP} improved the algorithm of~\cite{BGS00} to
obtain an algorithm for scc decomposition that takes at most linear 
amount of symbolic steps. 
In this section we present an improved version of the algorithm of~\cite{GPP} 
that improves the constants of the number of linear symbolic steps required.
In Section~\ref{subsec:main-idea} we present the improved
algorithm and correctness, and some further technical details are
presented in Section~\ref{subsec:technical} of appendix.

\subsection{Improved algorithm and correctness}\label{subsec:main-idea}
We first describe the main ideas of the algorithm of~\cite{GPP} and 
then present our improved algorithm.
The algorithm of~\cite{GPP} improves the algorithm of~\cite{BGS00}  by 
maintaining the right order for forward sets.
The notion of \emph{spine-sets} and \emph{skeleton of a forward set} was
designed for this purpose.

\smallskip\noindent{\bf Spine-sets and skeleton of a forward set.} 
Let $G=(S,E)$ be a directed graph. 
Consider a finite path $\tau=(s_0, s_1, \ldots, s_\ell)$, such that for all 
$0 \leq i \leq \ell-1$ we have $(s_i,s_{i+1}) \in E$. 
The path is \emph{chordless} if for all $0 \leq i < j \leq \ell$ such that 
$j-i > 1$, there is no edge from $s_i$ to $s_j$.
Let $U \subseteq S$. The pair $(U,s)$ is a \emph{spine-set} of $G$ iff
$G$ contains a chordless path whose set of states is $U$ that 
ends in $s$. 
%%The state $s$ is called the \emph{spine-anchor} of the spine-set $(U,s)$.
For a state $s$, let $\FW(s)$ denote the set of states that is reachable 
from $s$ (i.e., reachable by a forward search from $s$). 
The set $(U,t)$ is a \emph{skeleton of} $\FW(s)$ iff $t$ is a state in 
$\FW(s)$ whose distance from $s$ is maximum and $U$ is the set of states 
on a shortest path from $s$ to $t$. 
The following lemma was shown in~\cite{GPP} establishing relation of skeleton 
of forward set and spine-set.

\begin{lemma}[\cite{GPP}]
Let $G=(S,E)$ be a directed graph, and let $\FW(s)$ be the forward set of 
$s \in S$. 
The following assertions hold: (1) If $(U,t)$ is a skeleton of the forward-set 
$\FW(s)$, then $U \subseteq \FW(s)$.
(2) If $(U,t)$ is a skeleton of $\FW(s)$, then $(U,t)$ is a spine-set in $G$.
\end{lemma}

\noindent{\bf The intuitive idea of the algorithm.} 
The algorithm of~\cite{GPP} is a recursive
algorithm,  and in every recursive call the scc of a state $s$ is determined by 
computing $\FW(s)$, and then identifying the set of states in $\FW(s)$ having a 
path to $s$.
The choice of the state to be processed next is guided by the implicit inverse 
order associated with a possible spine-set. 
This is achieved as follows: whenever a forward-set $\FW(s)$ is computed, a 
skeleton of such a forward set is also computed. 
The order induced by the skeleton is then used for the subsequent computations.
Thus the symbolic steps performed to compute $\FW(s)$ are distributed over
the scc computation of the states belonging to a skeleton of $\FW(s)$. 
The key to establish the linear complexity of symbolic steps is the amortized
analysis. 
We now present the main procedure \sccfind\ and the main sub-procedure 
\skelfwd\  of the algorithm from~\cite{GPP}.

\smallskip\noindent{\bf Procedures \sccfind\ and \skelfwd.} 
The main procedure of the algorithm is \sccfind\ that calls \skelfwd\ 
as a sub-procedure. 
The input to \sccfind\ is a graph $(S,E)$ and $(A,B)$, where either 
$(A,B)=(\emptyset,\emptyset)$ or $(A,B)=(U,\set{s})$, where $(U,s)$ is a 
spine-set.
If $S$ is $\emptyset$, then the algorithm stops. 
Else, (a) if $(A,B)$ is $(\emptyset,\emptyset)$, then the procedure 
picks an arbitrary $s$ from $S$ and proceeds;
(b) otherwise, the sub-procedure \skelfwd\ is invoked to compute the forward 
set of $s$ together with the skeleton  $(U',s')$ of such a forward set. 
The \sccfind\ procedure has the following local variables: 
$\fw,\newset,\newstate$ and $\scc$.
The variable $\fw$ that maintains the forward set, whereas
$\newset$ and $\newstate$ maintain $U'$ and $\set{s'}$, respectively. 
The variable $\scc$ is initialized to $s$, and then augmented with the 
scc containing $s$. The partition of the scc's is updated and finally
the procedure is recursively called over:
\begin{enumerate}
\item the subgraph of $(S,E)$ is induced by $S \setminus \fw$ and the 
spine-set of such a subgraph obtained from $(U,\set{t})$ 
by subtracting $\scc$; %%\mynote{Spine set is (U,\{ t \} ) or (U,t) ?}
\item the subgraph of $(S,E)$ induced by $\fw\setminus \scc$ and the spine-set
of such a subgraph obtained from $(\newset,\newstate)$ by subtracting 
$\scc$.
\end{enumerate}
The \skelfwd\ procedure takes as input a graph $(S,E)$ and a state $s$, 
first it computes the forward set $\FW(s)$,
and second it computes the skeleton of the forward set.
The forward set is computed by symbolic breadth first search, and 
the skeleton is computed with a stack. 
The detailed pseudocodes are in the following subsection.
We will refer to this algorithm of~\cite{GPP} as \sym.
The following result was established in~\cite{GPP}: for the 
proof of the constant 5, refer to the appendix of~\cite{GPP} and 
the last sentence explicitly claims that every state is charged 
at most 5 symbolic steps. 

\begin{theorem}[\cite{GPP}]\label{thrm-cgp}
Let $G=(S,E)$ be a directed graph. 
The algorithm \sym\ correctly computes the scc decomposition 
of $G$ in $\min\set{5\cdot |S|, 5 \cdot D(G) \cdot N(G) + N(G)}$ symbolic steps,
where $D(G)$ is the diameter of $G$, and $N(G)$ is the number of 
scc's in $G$.
\end{theorem}

\smallskip\noindent{\bf Improved symbolic algorithm.} 
We now present our improved symbolic scc algorithm and refer to the algorithm 
as \isym. 
Our algorithm mainly modifies the sub-procedure \skelfwd.
The improved version of \skelfwd\ procedure takes an additional input 
argument $Q$, and returns an additional output argument that is stored as a 
set $P$ by the calling \sccfind\ procedure. 
The calling function passes the set $U$ as $Q$.
The way the output $P$ is computed is as follows: at the end of the forward search 
we have the following assignment: $P:= \fw \cap Q$. 
After the forward search, the skeleton of the forward set is computed 
with the help of a stack. 
The elements of the stacks are sets of states stored in the 
forward search. The spine set computation is similar to \skelfwd,
the difference is that when elements are popped of the stack, 
we check if there is a non-empty intersection with $P$, if so, 
we break the loop and return. %%%\mynote{??}
Moreover, for the backward searches in \sccfind\ we initialize $\scc$ by $P$ rather than $s$. 
We refer to the new sub-procedure as \iskelfwd\ (detailed pseudocode in 
the following subsection).

\smallskip\noindent{\bf Correctness.} 
Since $s$ is the last element of the spine set $U$, and $P$ is the intersection 
of a forward search from $s$ with $U$, it means that all elements of $P$ are 
both reachable from $s$ (since $P$ is a subset of $\FW(s)$) and can reach $s$ 
(since $P$ is a subset of $U$). 
It follows that $P$ is a subset of the scc containing $s$.
Hence not computing the spine-set beyond $P$ does not change the future 
function calls, i.e., the value of $U'$, since the omitted parts of $\newset$ 
are in the scc containing $s$.
The modification of starting the backward search from $P$ does not change 
the result, since $P$ will anyway be included in the backward search. 
So the \isym\ algorithm gives the same result as \sym, and the correctness 
follows from Theorem~\ref{thrm-cgp}.

\smallskip\noindent{\bf Symbolic steps analysis.} 
We present two upper bounds on the number of symbolic steps of the algorithm. 
Intuitively following are the symbolic operations that need to be 
accounted for: (1) when a state is included in a spine set for the first time in 
\iskelfwd\ sub-procedure which has two parts:
the first part is the forward search and the second part is 
computing the skeleton of the forward set; 
(2) when a state is already in a spine set and is found in forward search of 
\iskelfwd\ and (3) the backward search for determining the scc.
We now present the number of symbolic steps analysis for \isym.

\begin{enumerate}
\item There are two parts of \iskelfwd, (i)~a forward search and 
(ii)~a backward search for skeleton computation of the forward set. 
For the backward search, we show that the number of steps performed 
equals the size of $\newset$ computed. One key idea of the analysis is the proof where we 
show that a state becomes part of spine-set at most once, as compared to the algorithm 
of~\cite{GPP} where a state can be part of spine-set at most twice. Because, when it is already part 
of a spine-set, it will be included in $P$ and we stop the computation of spine-set when an 
element of $P$ gets included. We now split the analysis in two cases: (a) states that 
are included in spine-set, and (b)~states that are not included in spine-set.
\begin{enumerate}
\item We charge one symbolic step for the backward search of \iskelfwd\ (spine-set computation) 
to each element when it first gets inserted in a spine-set. For the forward search, we see that 
the number of steps performed is the size of spine-set that would have been computed if we did 
not stop the skeleton computation. But by 
stopping it, we are only omitting states that are part of the scc. Hence we charge one symbolic 
step to each state getting inserted into spine-set for the first time and each state of the scc. 
Thus, a state getting inserted in a spine-set is charged two symbolic steps 
(for forward and backward search) of \iskelfwd\ the first time it is inserted.

\item A state not inserted in any spine-set is charged one symbolic step for backward search which 
determines the scc.
\end{enumerate}
Along with the above symbolic steps, one step is charged to each state for the forward 
search in \iskelfwd\ at the time its scc is being detected.
Hence each state gets charged at most three symbolic steps.
%(states in spine set charged two 
%when included in the spine-set, states not in spine-set charged one for backward search 
%determining scc and one to each state in forward search of \iskelfwd)
Besides, for computing $\newstate$, one symbolic step is required per scc found.
Thus the total number of symbolic steps is bounded by $3 \cdot |S| +  N(G)$, 
where $N(G)$ is the number of scc's of $G$.

\item Let $D^*$ be the sum of diameters of the scc's in a $G$. 
Consider a scc with diameter $d$. 
In any scc the spine-set is a shortest path, and hence the size of the 
spine-set is bounded by $d$. 
Thus the three symbolic steps charged to states in spine-set contribute to at most
$3\cdot d$ symbolic steps for the scc. 
Moreover, the number of iterations of forward search of \iskelfwd\ charged 
to states belonging to the scc being computed are at most $d$. And the number 
of iterations of the backward search to compute the scc is also at most $d$.
Hence, the two symbolic steps charged to states not in any spine-set also contribute 
at most $2\cdot d$ symbolic steps for the scc.
Finally, computation of $\newset$ takes one symbolic step per scc.
Hence we have $5\cdot d+1$ symbolic steps for a scc with diameter $d$. 
We thus obtain an upper bound of $5D^*+N(G)$ symbolic steps.
\end{enumerate}
It is straightforward to argue that the number of symbolic steps of 
\isccfind\ is at most the number of symbolic steps of \sccfind.
The detailed pseudocode and technical details of the running time analysis is presented in 
the appendix.
%%following subsection. %%appendix.

\begin{theorem}\label{thrm_sym_scc}
Let $G=(S,E)$ be a directed graph. 
The algorithm \isym\ correctly computes the scc decomposition 
of $G$ in $\min\set{3 \cdot |S| + N(G), 5 \cdot D^*(G) + N(G)}$ symbolic steps,
where $D^*(G)$ is the sum of diameters  of the scc's of $G$, and $N(G)$ is the number of 
scc's in $G$.
\end{theorem}

\begin{remark}
Observe that in the worst case \sccfind\ takes $5 \cdot n$ symbolic steps, whereas
\isccfind\ takes at most $4 \cdot n$ symbolic steps. 
Thus our algorithm improves the constant of the number of linear symbolic steps 
required for symbolic scc decomposition.
\end{remark}

\section{Experimental Results}
In this section we present our experimental results. We first present the
results for symbolic algorithms for MDPs with B\"uchi objectives and then 
for symbolic scc decomposition.

\smallskip\noindent{\bf Symbolic algorithm for MDPs with B\"uchi objectives.} 
We implemented all the symbolic algorithms (including the classical one) and 
ran the algorithms on randomly generated graphs. 
If we consider arbitrarily randomly generated graphs, then in most cases it 
gives rise to trivial MDPs.
Hence we generated more structured MDP graphs. First we generated a large number of
MDPs and as a first step chose the MDP graphs where all the algorithms required 
large number of symbolic steps, and then generated large number of MDP graphs 
randomly by small perturbations of the graphs chosen in the first step.
%%Hence we generated large number of MDP graphs randomly, first chose the ones 
%where all the algorithms required the most number of symbolic steps, and then 
%%considered random graphs obtained by small uniform perturbations of them.
Our results of average symbolic steps required are shown in Table~\ref{tab1} 
and show that the new algorithms perform significantly better than 
the classical algorithm.
The running time comparison is given in Table~\ref{tab1r}.
%%We also observed that the \wl\ algorithm performs better in graphs with many scc's.

\begin{table}
\begin{center}
\begin{tabular}{|c|c|c|c|c|}
\hline
Number of states  & Classical & \simpa\   & \oimpa\  & \siwl\  \\
\hline
5000 &  30731 & 3478 & 3898 & 3573 \\ 
\hline
10000  & 103977 & 6622 & 7490 & 6815 \\ 
\hline
20000   & 306015 & 12010 & 13212 & 13687 \\ 
\hline
\hline
\end{tabular}
\end{center}
\caption{The average symbolic steps required by symbolic algorithms for MDPs
with B\"uchi objectives.}\label{tab1}
\end{table}

\begin{table}
\begin{center}
\begin{tabular}{|c|c|c|c|c|}
\hline
Number of states  & Classical & \simpa\   & \oimpa\  & \siwl\  \\
\hline
5000 &  78.8 & 9.7 & 10.2 & 10.8 \\ 
\hline
10000  & 563.7 & 40.3 & 43.0 & 46.1 \\ 
\hline
20000   & 3974.4 & 186.4 & 192.3 & 217.4 \\ 
\hline
\hline
\end{tabular}
\end{center}
\caption{The average running time required in sec by symbolic algorithms for MDPs
with B\"uchi objectives.}\label{tab1r}
\end{table}

\noindent{\bf Symbolic scc computation.} We implemented the symbolic 
scc decomposition algorithm from~\cite{GPP} and our new symbolic algorithm.
A comparative study of the algorithm of~\cite{GPP} and the algorithm 
of~\cite{BGS00} was done in~\cite{FabioPersonal}, and it was found that the 
performances were comparable.
Hence we only perform the comparison of the algorithm of~\cite{GPP} and our
new algorithm. 
We ran the algorithms on randomly generated graphs. 
Again arbitrarily randomly generated graphs in many cases gives rise to 
graphs that are mostly disconnected or completely connected. 
Hence we generated random graphs by first constructing a topologically
sorted order of the scc's and then adding edges randomly respecting the 
topologically sorted order.
Our results of average symbolic steps are shown in Table~\ref{tab2} and 
shows that our new algorithm performs better 
(around 15\% improvement over the algorithm of~\cite{GPP}).
The running time comparison is shown in Tab~\ref{tab2r}.

\begin{table}[ht]
\begin{center}
\begin{tabular}{|c|c|c|c|c|}
\hline
Number of states  & Algorithm from~\cite{GPP} & Our Algorithm & Percentage Improvement\\
\hline
10000 & 1043  & 878  & 15.83 \\ 
\hline
25000  & 2649 & 2264  & 14.53 \\ 
\hline
50000   & 6299 & 5394 & 14.36\\ 
\hline
\hline
\end{tabular}
\end{center}
\caption{The average symbolic steps required for scc computation.}\label{tab2}
\end{table}

\begin{table}
\begin{center}
\begin{tabular}{|c|c|c|c|c|}
\hline
Number of states  & Algorithm from~\cite{GPP} & Our Algorithm & Percentage Improvement\\
\hline
10000 & 7.7  & 6.3  & 17.53 \\ 
\hline
25000  & 48.3 & 40.0  & 16.98 \\ 
\hline
50000   & 180.8 & 152.5 & 15.67\\ 
\hline
\hline
\end{tabular}
\end{center}
\caption{The average running time required in sec for scc computation.}\label{tab2r}
\end{table}

In all cases, our implementations were the basic implementation of the 
algorithms, and more optimized implementations would lead to improved 
performance results.
The source codes, sample examples for the experimental results, and other 
details of the implementation are available at \url{http://www.cs.cmu.edu/~nkshah/SymbolicMDP}.

\section{Conclusion}
In this work we considered a core problem of probabilistic verification
which is to compute the set of almost-sure winning states in MDPs 
with B\"uchi objectives. 
We presented the first symbolic sub-quadratic algorithm for the problem, 
and also a new symbolic sub-quadratic algorithm (\iwl\ algorithm). 
As compared to all previous algorithms which idnetify the almost-sure
winning states upon termination, the \iwl\ algorithm can potentially discover 
almost-sure winning states in intermediate steps as well.
Finally we considered another core graph theoretic problem in verification
which is the symbolic scc decomposition problem.
We presented an improved algorithm for the problem. 
The previous best known algorithm for the problem required $5\cdot n$ symbolic
steps in the worst case and our new algorithm takes at most $4\cdot n$
symbolic steps, where $n$ is the number of states of the graph.
Our basic implementation shows that our new algorithms perform favorably
over the old algorithms.
Optimized implementations of the new algorithms and detailed experimental 
studies would be an interesting direction for future work.

\medskip\noindent{\bf Acknowledgements.} We thank Fabio Somenzi for sharing 
the facts about the performance comparison of the algorithm of~\cite{BGS00}
and the algorithm of~\cite{GPP}.
We thank anonymous reviewers for many helpful comments that improved the 
presentation of the paper.

\clearpage
\section{Appendix}

\newcommand{\stack}{\mathsf{stack}}
\newcommand{\pick}{\mathsf{pick}}
\newcommand{\sccpartition}{\mathsf{SCCPartition}}
\newcommand{\push}{\mathsf{Push}}
\newcommand{\pop}{\mathsf{Pop}}

\subsection{Technical details of improved symbolic scc algorithm}\label{subsec:technical}
%%\subsection{Pesudocodes of \sccfind\ and \isccfind}
The pseudocode of \sccfind\ is formally given as 
Algorithm~\ref{algorithm:sccfind}. 
The correctness analysis and the analysis of the number of 
symbolic steps is given in~\cite{GPP}.
The pseudocode of \isccfind\ is formally given as  Algorithm~\ref{algorithm:isccfind}. 
The main changes of \isccfind\ from \sccfind\ are as follows: 
(1) instead of \skelfwd\ the algorithm \isccfind\ calls procedure \iskelfwd\ 
that returns an additional set $P$ and \iskelfwd\ is invoked with an additional
argument that is $U$; 
(2) in line 4 of \isccfind\ the set $\scc$ is initialized to $P$ 
instead of $s$.
The main difference of \iskelfwd\ from \skelfwd\ is as follows: 
(1) the set $P$ is computed in line 4 of \iskelfwd\ as $\fw \cap Q$,
where $Q$ is the set passed by \isccfind\ as the argument; and 
(2) in the while loop it is checked if the element popped intersects 
with $P$ and if yes, then the procedure breaks the while loop.
The correctness argument from the correctness of \sccfind\ is already
shown in Section~\ref{subsec:main-idea}.

\begin{algorithm}[!ht]
\caption{\bf \sccfind}
\label{algorithm:sccfind}
{ 
\begin{tabbing}
aa \= aa \= aa \= aa \= aa \= aa \= aa \= aa \= aa \= aa \= aa \= aa \kill
\> {\bf Input:}  $(S,E, \langle U,s \rangle)$, i.e., a graph $(S,E)$ with spine set $(U,s)$. \\
\> {\bf Output:} $\sccpartition$ i.e. the set of SCCs of the graph $(S,E)$ \\

\> {\bf Initialize} $\sccpartition:=\emptyset$; $\fw:=\emptyset$; \\
\> 1. {\bf if}  $(S = \emptyset)$ {\bf then} \\
\>\>  1.1 {\bf return}; \\
\> 2. {\bf if } $(U = \emptyset)$ {\bf then} \\
\>\> 2.1 $s := \pick(S)$ \\ %%(*WHAT IS pick*) \\
\> 3. $\langle \fw, \newset, \newstate  \rangle :=$ \skelfwd $(S,E,s)$ \\
\> 4. $\scc = s$  \\ %%(*is it $s$ or $\set{s}$* \bf{This is just s since s itself is a singleton set rather than a state})\\
\> 5. {\bf while} $(((\pre(\scc) \cap \fw) \setminus \scc) \neq \emptyset)$ {\bf do }\\
\>\> 5.1 $\scc := \scc \cup (\pre(\scc) \cap \fw)$ \\
\> 6. $\sccpartition := \sccpartition \cup \{\scc\}$ \\
(Recursive call on $S \setminus \fw$) \\
\> 7. $S^{\prime} := S \setminus \fw$ \\
\> 8. $E^{\prime} := E \cap (S^{\prime} \times S^{\prime})$ \\
\> 9. $U^{\prime} := U \setminus \scc $ \\
\> 10. $s^{\prime} := \pre(\scc \cap U) \cap (S \setminus \scc)$ \\ %%(** Check that $s^{\prime}$ is singleton ***)\\
\> 11. $\sccpartition := \sccpartition \cup$ \sccfind $(S^{\prime},E^{\prime}, \langle U^{\prime},s^{\prime} \rangle) $\\
(Recursive call on $\fw \setminus \scc$) \\
\> 12. $S^{\prime} := \fw \setminus \scc$ \\
\> 13. $E^{\prime} := E \cap (S^{\prime} \times S^{\prime})$ \\
\> 14. $U^{\prime} := \newset \setminus \scc $\\
\> 15. $s^{\prime} := \newstate \setminus \scc $\\
\> 16. $\sccpartition := \sccpartition \cup$ \sccfind $(S^{\prime},E^{\prime}, \langle U^{\prime},s^{\prime} \rangle) $\\
\> 17. Return $\sccpartition$
\end{tabbing}
}
{ 
\begin{tabbing}
aa \= aa \= aaa \= aaa \= aaa \= aaa \= aaa \= aaa \kill
{\bf Procedure \skelfwd} \\
\>{\bf Input:} $(S,E,s)$, i.e., a graph $(S,E)$ with a state $s \in S$. \\
\>{\bf Output:}  $\langle \fw,\newset,\newstate \rangle$, \\ 
\>\> i.e. forward set $\fw$, new spine-set $\newset$ and $\newstate \in \newset$ \\

\> 1. Let $\stack$ be an empty stack of sets of nodes \\
\> 2. $L := s$ \\
\> 3. {\bf while} $(L \neq \emptyset)$ {\bf do} \\
\>\> 3.1 $\push(\stack,L)$ \\
\>\> 3.2 $\fw := \fw \cup L$ \\
\>\> 3.3 $L := \post(L) \setminus \fw$ \\

\> 4. $L := \pop(\stack)$ \\
\> 5. $\newset := \newstate := \pick(L)$ \\ %%(*WHAT IS PICK*)\\
\> 6. {\bf while} $(\stack \neq \emptyset)$ {\bf do} \\
\>\> 6.1 $L := \pop(\stack)$ \\
\>\> 6.2 $\newset := \newset \cup \pick(\pre(\newset) \cap L)$ \\

7. {\bf return} $\langle \fw,\newset,\newstate \rangle$ \\
\end{tabbing}
}
\end{algorithm}

\begin{algorithm}[t]
\caption{\bf \isccfind}
\label{algorithm:isccfind}
{ 
\begin{tabbing}
aa \= aa \= aa \= aa \= aa \= aa \= aa \= aa \= aa \= aa \= aa \= aa \kill
\> {\bf Input:}  $(S,E, \langle U,s \rangle)$, i.e., a graph $(S,E)$ with spine set $(U,s)$. \\
\> {\bf Output:} $\sccpartition$, the set of SCCs of the graph $(S,E)$ \\
\> {\bf Initialize} $\sccpartition:=\emptyset$; $\fw:=\emptyset$; \\
\> 1. {\bf if}  $(S = \emptyset)$ {\bf then} \\
\>\>  1.1 {\bf return}; \\
\> 2. {\bf if } $(U = \emptyset)$ {\bf then} \\
\>\> 2.1 $s := \pick(S)$ \\ %%(*WHAT IS pick*) \\
\> 3. $\langle \fw, \newset, \newstate, P  \rangle :=$ \iskelfwd $(S,E,U,s)$ \\
\> 4. $\scc = P$ \\
\> 5. {\bf while} $(((\pre(\scc) \cap \fw) \setminus \scc) \neq \emptyset)$ {\bf do }\\
\>\> 5.1 $\scc := \scc \cup (\pre(\scc) \cap \fw)$ \\
\> 6. $\sccpartition := \sccpartition \cup \{\scc\}$ \\
(Recursive call on $S \setminus \fw$) \\
\> 7. $S^{\prime} := S \setminus \fw$ \\
\> 8. $E^{\prime} := E \cap (S^{\prime} \times S^{\prime})$ \\
\> 9. $U^{\prime} := U \setminus \scc $ \\
\> 10. $s^{\prime} := \pre(\scc \cap U) \cap (S \setminus \scc)$ \\  %%(*** Check $s^{\prime}$ singleton **)\\
\> 11. $\sccpartition := \sccpartition \cup$ \isccfind $(S^{\prime},E^{\prime}, \langle U^{\prime},s^{\prime} \rangle) $\\
(Recursive call on $\fw \setminus \scc$) \\
\> 12. $S^{\prime} := \fw \setminus \scc$ \\
\> 13. $E^{\prime} := E \cap (S^{\prime} \times S^{\prime})$ \\
\> 14. $U^{\prime} := \newset \setminus \scc $\\
\> 15. $s^{\prime} := \newstate \setminus \scc $\\
\> 16. $\sccpartition := \sccpartition \cup $ \isccfind $(S^{\prime},E^{\prime}, \langle U^{\prime},s^{\prime} \rangle) $\\
\end{tabbing}
}
{ 
\begin{tabbing}
aa \= aa \= aaa \= aaa \= aaa \= aaa \= aaa \= aaa \kill
{\bf Procedure \iskelfwd} \\
\>{\bf Input:} $(S,E,Q,s)$, i.e., a graph $(S,E)$ with a set $Q$ and a state $s \in S$. \\
\>{\bf Output:}  $\langle \fw,\newset,\newstate, P \rangle$ \\

\> 1. Let $\stack$ be an empty stack of sets of nodes \\
\> 2. $L := s$ \\
\> 3. {\bf while} $(L \neq \emptyset)$ {\bf do} \\
\>\> 3.1 $\push(\stack,L)$ \\
\>\> 3.2 $\fw := \fw \cup L$ \\
\>\> 3.3 $L := \post(L) \setminus \fw$ \\
\> 4. $P := \fw \cap Q$ \\
\> 5. $L := \pop(\stack)$ \\
\> 6. $\newset := \newstate := \pick(L)$ \\ %%(*WHAT IS PICK*)\\
\> 7. {\bf while} $(\stack \neq \emptyset)$ {\bf do} \\
\>\> 7.1 $L := \pop(\stack)$ \\
\>\> 7.2 {\bf if} $(L \cap P \neq \emptyset)$ {\bf then} \\
\>\>\> 7.2.1 {\bf break while loop} \\
\>\> 7.3 {\bf else} $\newset := \newset \cup \pick(\pre(\newset) \cap L)$ \\

8. {\bf return} $\langle \fw,\newset,\newstate, P \rangle$ \\
\end{tabbing}
}
\end{algorithm}

%%\smallskip\noindent{\bf Correctness.} {\bf $\clubsuit$ Krish to MN: do we need to say anything more about correctness.$\clubsuit$.} 

\smallskip\noindent{\bf Symbolic steps analysis.}
We now present the detailed symbolic steps analysis of the algorithm.
As noted in Section \ref{sec:symimpa}, common symbolic operations on a set of states are 
$\Pre$, $\Post$ and $\CPre$. 
We note that these operations involve symbolic sets of $2\cdot \log(n)$ variables, as 
compared to symbolic sets of $\log(n)$ variables required for operations such as union, 
intersection and set difference. 
Thus only $\Pre$, $\Post$ and $\CPre$ are counted as symbolic steps, as done in~\cite{GPP}.
The total number of other symbolic operations is also $O(|S|)$. 
We note that only lines 5 and 10 of \isccfind\ and lines 3.3 and 7.3 of \iskelfwd\  
involve $\Pre$ and $\Post$ operations.

\smallskip\noindent In the following, we charge the costs of these lines to states in 
order to achieve the $3 \cdot |S|+N(G)$ bound for symbolic steps. 
We define subspine-set as $\newset$ returned by \iskelfwd\ and 
show the following result.

\begin{lemma}\label{lemma:subspine}
For any spine-set $U$ and its end vertex $u$, $T$ is a subspine-set iff 
$U \setminus T \subseteq \scc(u)$. 
\end{lemma}
\begin{proof}
Note that while constructing a subspine-set $T$, we stop the construction when we find any state $v \in \fw \cap U$ 
from the spine set.
%%in a subspine-set. 
Now clearly since $v \in U$, there is a path from $v$ to $u$.
Also, since we found this state in $FW(u)$, there is a path from $u$ to $v$. Hence, $v \in \scc(u)$. 
Also, each state that we are omitting by stopping construction of $T$ has the property that there is a path from $u$ to that state and a path from that state to $v$.
This implies that all the states we are omitting in construction of $T$ are in $\scc(u)$.
\hfill \qed
\end{proof}

Note that since we pass $\newset \setminus \scc$ in the subsequent call to \isccfind, 
it will actually be a spine set for the reduced problem. 
In the following lemma we show that any state can be part of subspine-set at most
once, as compared to twice in the \sccfind\ procedure in~\cite{GPP}. 
This lemma is one of the key points that lead to the improved analysis of 
symbolic steps required. 
%%analysis.

\begin{lemma}\label{lemma:atmostonce}
Any state $v$ can be part of subspine-set at most once.
\end{lemma}
\begin{proof}
In \cite{GPP}, the authors show that any state $v$ can be included in spine sets at most twice in \skelfwd. 
The second time the state $v$ is included is in line 6 of \skelfwd\ when the $\scc(v)$ of the state is to be found. 
In contrast, \iskelfwd\ checks intersection of the subspine-set being constructed with the set $P$ that contains the states of $\scc(v)$ which are already in a subspine-set.
When this happens, it stops the construction of the subspine-set. 
Now if $v$ is already included in the subspine-set, then it will be part of $P$ and would not be included in subspine-set again.
Hence, $v$ can be part of subspine-set at most once. 
\hfill \qed
\end{proof}

\begin{lemma}
\label{lemma:nonspine}
States added in $\scc$ by iteration of line 5 of \isccfind\ are exactly the states which are not part of any subspine-set.
\end{lemma}
\begin{proof}
We see that in line 5 of \isccfind, we start from $\scc = P$ and then we find the $\scc$ by backward search. 
Also, $P$ has all the states from $\scc$ which are part of any subspine-set. 
Hence, the extra states that are added in $\scc$ are states which are never included in a subspine-set.
\hfill \qed
\end{proof}

\smallskip\noindent{\bf Charging symbolic steps to states.} We now consider three cases 
to charge symbolic steps to states and scc's.
\begin{enumerate} 
\item 
\emph{Charging states included in subspine-set.}
First, we see that the number of times the loop of line 3 in \iskelfwd\ is executed
is equal to the size of the spine set that \skelfwd\ would have computed.
Using Lemma \ref{lemma:subspine}, we can charge one symbolic step to each state of the subspine-set and each state of the $\scc$ that is being computed.
Now, the number of times line 7.3 of \iskelfwd\ is executed equals the size of subspine-set that is computed. 
Hence, we charge one symbolic step to each state of subspine-set for this line.

\smallskip\noindent Now we summarize the symbolic steps charged to each state which is part of some subspine-set. First time when a state gets into a subspine-set, it is charged two steps, one for line 3.3 and one for line 7.3 of \iskelfwd. If its $\scc$ is not found in the same call to \isccfind, then it comes into action once again when its $\scc$ is being found. By Lemma \ref{lemma:atmostonce}, it is never again included in a subspine set. Hence in this call to \iskelfwd, it is only charged one symbolic step for line 3.3 and none for line 7.3 as line 7.3 is charged to states that become part of the newly constructed subspine-set. Also because of Lemma \ref{lemma:nonspine}, since this state is in a subspine-set, it is not charged anything for line 5 of \isccfind.
Hence, a state that occurs in any subspine-set is charged at most three symbolic steps.

\item \emph{Charging states not included in subspine-set.}
For line 5 of \isccfind, the number of times it is executed is the number of states that are added to $\scc$ after initialization to $\scc = P$.
Using Lemma~\ref{lemma:nonspine}, we charge one symbolic step to each state of this $\scc$ that is never a part of any subspine-set.
Also, we might have charged one symbolic step to such a state for line 3.3 of \iskelfwd\ when we called it.
Hence, each such state is charged at most two symbolic steps.

\item \emph{Charging $\scc$s.}
For line 10 of \isccfind, we see that it is executed only once in a call to \isccfind\ that computes a $\scc$. 
Hence, the total number of times line 10 is executed equals $N(G)$, the number of $\scc$s of the graph.
Hence, we charge each $\scc$ one symbolic step for this line.

\end{enumerate}

\smallskip\noindent 
The above argument shows that the number of symbolic steps that the algorithm \isccfind\ requires is at most $3 \cdot |S|+N(G)$.
This completes the formal proof of Theorem~\ref{thrm_sym_scc}.

\begin{comment}

\isccfind $(S,E, \langle U,s \rangle)$ \\
if $S = \phi$ \\
then return; \\

if $U = \phi$ \\
then $s = pick(V)$ \\

$\langle \fw,\newset,\newstate,P \rangle$ = \iskelfwd $(S,E,U,s)$ \\

$\scc = P$ \\
while $((\pre(\scc) \cap \fw) \setminus \scc) \neq \emptyset$ \\
do \\
$\scc := \scc \cup (\pre(\scc) \cap \fw)$ \\

$SCC-Parition = SCC-Partition \cup \{\scc\}$ \\

$S^{\prime} = S \setminus \fw$ \\
$E^{\prime} = E \cap S^{\prime} \times S^{\prime}$ \\
$U^{\prime} = U \setminus SCC $ \\
$s^{\prime} = \pre(\scc \cap U) \cap (S \setminus \scc )$ \\
$SCCFIND(S^{\prime},E^{\prime}, \langle U^{\prime},s^{\prime} \rangle) $\\

$S^{\prime} = \fw \setminus \scc$ \\
$E^{\prime} = E \cap S^{\prime} \times S^{\prime}$ \\
$U^{\prime} = \newset \setminus \scc $\\
$s^{\prime} = \newstate \setminus \scc $\\
\isccfind $(S^{\prime},E^{\prime}, \langle U^{\prime},s^{\prime} \rangle) $\\

\subsection{ImprovedSKEL-FORWARD pseudocode}
\iskelfwd $(S,E,Q,s)$ \\
Let $\stack$ be an empty stack of sets of nodes \\

$L = s$ \\
while $(L \neq \emptyset)$ do \\
$push(stack,L)$ \\
$\fw = \fw  \cup L$ \\
$L = post(L) \setminus \fw $ \\

$P = \fw \cap Q$ \\
$L = pop(stack)$ \\
$\newset = \newstate = pick(L)$ \\
while $\stack \neq \phi$ do \\
$L = pop(\stack)$ \\
if $(L \cap P) \neq \emptyset$ then \\
break\\
$\newset = \newset \cup pick(\pre(\newset) \cap L)$ \\
return $\langle \fw,\newset,\newstate,P \rangle$ \\

\end{comment}

We now present an example that presents a family of graphs with $k\cdot n$
states, where the \sccfind\ algorithm takes almost $5\cdot k \cdot n$ 
symbolic steps, whereas the \isccfind\ algorithm takes at most
$3\cdot k \cdot n + n$ symbolic steps.

\begin{example}
Let $k,n \in \Nats$.
Consider a graph with $k\cdot n$ states such that the states are 
numbered from $1$ to $k\cdot n$.
The edges are as follows: (1)~for all states $1\le i \le k\cdot n-1$, there is an edge $(i,i+1)$ 
(i.e, the states are all in a line); and 
(2)~for all $1\le i \le n$, there is an edge from state $k\cdot i$ to state $(i-1)\cdot k+1$.
We will show that the \sccfind\ algorithm requires roughly $(5\cdot k-1)\cdot n$ symbolic steps
on this graph.
Note the the number of scc's in this graph is $n$, and hence by Theorem~\ref{thrm_sym_scc}
the \isccfind\ algorithm takes at most $3\cdot k\cdot n + n$ symbolic steps. 

We now analyze the symbolic steps required by the \sccfind\ algorithm, and our analysis 
is in two steps.
\begin{enumerate}
\item \emph{Step~1.}
In the beginning, starting from the state~1 of the graph, 
the algorithm performs a forward and backward search to find the spine set. 
This will have a cost of two symbolic steps per state (one symbolic step while going forward, 
and one symbolic step while going backward), except for the first vertex which gets 
charged only one symbolic step (it does not get charged while going backwards). 
This gives a cost of $2\cdot k\cdot n-1$ symbolic steps. 
After this, the first scc will be found with an additional cost of only $k$, and the
first discovered scc consists of states $\set{1,2,\ldots,k}$.
Hence the total symbolic steps required is $2\cdot k\cdot n-1 +k$.

\item \emph{Step~2.} After Step~1, 
the algorithm will start finding scc's from the end of the spine set. 
So consider the set of last $k$ states from $k\cdot(n-1)+1$ to $n\cdot k$. 
The algorithm will pick the last state $n\cdot k$, 
find a spine set, consisting of the last $k$ states (states $k\cdot(n-1)+1$ to $n\cdot k$). 
This will have a cost of $2\cdot k-1$ symbolic steps ($k$ symbolic steps for the 
forward search, and $k-1$ symbolic steps for the backward search). 
After this, the algorithm will find the scc containing the last state $n\cdot k$ (i.e., the scc 
that consists of states $\set{k\cdot(n-1)+ 1,k\cdot(n-1)+ 2,\ldots,n\cdot k}$), and this 
takes $k$ symbolic steps as there are $k$ vertices in the scc.
Now Step~2 is repeated with state $k\cdot(n-1)$, and then repeated for state 
$k\cdot(n-2)$ and so on.
So the cost for every scc,  except the very first one, is $2k-1 + k = 3k-1$. 
Since there are $n$ scc's, the total number of symbolic steps required for Step~2 is
at least $(3\cdot k-1)\cdot (n-1)$.
\end{enumerate}
Hence the total symbolic steps required for the algorithm is at least
\[
2\cdot k\cdot n - 1 + k + (3\cdot k -1)\cdot (n-1)= (5\cdot k-1)\cdot n -2\cdot k.
\]
Note that with $k=n$, \sccfind\ takes at least $5\cdot n^2 -3\cdot n$ symbolic steps, whereas
the \isccfind\ takes at most $3 \cdot n^2 + n$ symbolic steps.
\hfill\qed
\end{example}

\end{document}